%% file: main.tex
\documentclass[a4paper,UKenglish,cleveref]{lipics-v2019}

\usepackage{graphicx}
\usepackage{url}
\usepackage{amsmath,amssymb,amsthm,latexsym}
\usepackage{algorithm}
\usepackage[noend]{algorithmic}
\algsetup{linenosize=\scriptsize}

\usepackage{multicol}
\usepackage{hyperref}

\usepackage{color}

\newtheorem{observation}[theorem]{Observation}

\renewcommand{\paragraph}[1]{\medskip\noindent{\bf #1}.\ }

\newcommand{\remove}[1]{}

\nolinenumbers

\title{Locally Solvable Tasks \\ and the Limitations of Valency Arguments}
\titlerunning{Locally Solvable Tasks}

\author{Hagit Attiya}{Computer Science Department, Technion}{hagit@cs.technion.ac.il}{0000-0002-8017-6457}{Supported by ISF grant 380/18.}

\author{Armando Casta\~neda}{Instituto de Matem\'aticas, UNAM}{armando.castaneda@im.unam.mx}{}{Supported by UNAM-PAPIIT project IN108720.}

\author{Sergio Rajsbaum}{Instituto de Matem\'aticas, UNAM}{rajsbaum@matem.unam.mx}{0000-0002-0009-5287}{Supported by UNAM-PAPIIT project IN106520.}

\authorrunning{Attiya, Casta\~neda and Rajsbaum}

\Copyright{Hagit Attiya, Armando Casta\~neda, Sergio Rajsbaum}



\begin{document}

\maketitle

\input{abstract}
\input{introduction}

\input{model}
\input{local-solvability}

\input{almost-global}

\input{FLP}
\input{discussion}

\paragraph{Acknowledgments}
We thank Ulrich Schmid and the reviewers  for many helpful comments.
We thank Dan Alistarh, James Aspnes, Faith Ellen, Rati Gelashvili and Leqi Zhu
for helpful conversations.
\bibliographystyle{abbrv}
\bibliography{biblio}

\end{document}

%% file: abstract.tex
\begin{abstract}
An elegant strategy for proving impossibility results
in distributed computing was introduced in
the celebrated FLP consensus impossibility proof.
This strategy is \emph{local} in nature as at each stage,
one configuration of a hypothetical protocol for consensus
is considered, together with future valencies of
possible extensions.
This proof strategy has been used in numerous situations related
to {consensus}, leading one to wonder why it has not been used
in  impossibility results of  two other well-known tasks:
\emph{set agreement} and \emph{renaming}.
This paper provides an explanation of why  impossibility proofs of these tasks have been of a global nature.
It shows that {a protocol can always solve such tasks locally},
in the following sense.
Given a configuration and all its future valencies, if a single
successor configuration is selected, then the protocol can
reveal all decisions in this branch of executions, satisfying
the task specification. 
This result is shown for both set agreement and renaming, implying that
there are no local impossibility proofs for these tasks.
\end{abstract}

%% file: introduction.tex

\section{Introduction}

An elegant strategy for proving impossibility results
in distributed computing was introduced in
the celebrated FLP consensus impossibility proof~\cite{FLP85}.
This strategy is \emph{local} in nature as at each stage,
one configuration of a hypothetical protocol for consensus is considered,
together with its future \emph{valencies}, namely,
the  decisions the protocol may reach from this configuration.
To apply it, one needs to consider only the interactions of pending
transitions at the configuration,
and analyze their commutativity properties.
This local nature makes the strategy very powerful and flexible, and has therefore
 been used in numerous situations related
to consensus (e.g.,~\cite{AT99,Aspnes1998,ACHP18,AttiyaC2008,ALS94,DDS87,H91,Her91,KeidarR03,LoH00,LAA87,MosesR2002}).

For this reason,
it would be desirable to be able to use a local strategy, in the style of FLP, to
prove impossibility results for two other important tasks:
\emph{$k$-set agreement}~\cite{Cha93}, an extension of consensus,
where processes may decide on up to $k$ different values,
and \emph{$M$-renaming}~\cite{ABDPR90},
where processes must pick distinct names from a given namespace of size $M$.
Existing impossibility proofs for these tasks
(e.g.,~\cite{AttiyaC2013,AttiyaP2016,BG93b,CR10,HerlihyRadv2013,HS99,SaksZ2000})
are based on topological invariant properties of final configurations of a protocol,
which are \emph{global} in nature,
namely, all final configurations are analyzed together to argue that
there is no protocol for the task.
For consensus,
these configurations are connected,
in the graph-theoretic sense. For set agreement and renaming,
higher-dimensional connectivity properties are proved.
Researchers have wondered why only global impossibility
proofs have been used for these tasks~\cite{AlistarhAEGZ2019}.

This paper provides an explanation of why the impossibility proof
strategies for set agreement and renaming  have been of a global nature.
It shows that one could not hope to prove that set agreement
and renaming are unsolvable through a local argument,
since \emph{they are solvable in a local sense}.
For a configuration $C$ of the protocol,
we denote by $\chi(C)$ all its successor configurations.
In a local FLP style of argument,
one selects a configuration $C'\in \chi(C)$,
based on the valencies of the configurations in $\chi(C)$.
The observation is that valencies can be assigned to $\chi(C)$,
such that for any chosen configuration $C'\in  \chi(C)$,
the protocol can reveal decisions in all final configurations
extending $C'$, such that  the decisions are consistent
both with the valencies and with the task specification.
Intuitively, a hypothetical protocol for set agreement or
renaming can ``hide'' its errors, if one inspects it only locally.

Intuitively,
the reason that a protocol can do this for set agreement and renaming,
and not for consensus,
is that the consensus specification is one-dimensional in nature,
so one can ``corner'' the protocol to reveal a configuration violating
agreement (assuming the protocol terminates).
Formally, it is always possible to find a \emph{bivalent} configuration
for consensus,
and it is impossible to locally solve consensus from such a configuration.
For set agreement and renaming, the protocol can ``move'' its errors around,
on a higher dimensional space,
without being cornered,
even if the protocol declares all its valencies.

In more detail, given a hypothetical full-information protocol
for either set agreement or renaming, we introduce the notion of  \emph{valency task}
 for set agreement and  for renaming. The inputs to such a task are
 the configurations $\chi(C)$  of the  protocol
after $\ell$ rounds, $\ell \geq 1$  (one round after  some configuration $C$).
For each configuration $C'\in\chi(C)$,
there is a valency, $val(C')$, specifying the  outputs of the protocol  on executions starting
in $C'$. The valency task is thus defined together by both  $\chi(C)$ and the valencies.
 A protocol \emph{solves the valency task locally} in $m\geq 1$ rounds,
 if starting on any  $C'\in\chi(C)$, after $m$ rounds it produces decisions
 that are consistent with the task specification (either set agreement or renaming),
 and additionally  \emph{complete},
that is, if a value $v\in val(C')$, then at least one process decides $v$ in at least one execution
starting in $C'$.
This captures the notion that the values promised by valencies are indeed decided.

We present the notions of valency task and local solvability in
Section~\ref{sec:valTaskLocSolv},
and define valency tasks for set agreement, {renaming and weak symmetry breaking, a task that is equivalent to renaming}.
We show in Section~\ref{sec-almost-global} that for both valency tasks,
set agreement and {weak symmetry breaking},
for any $\ell\geq 1$, the valency task is locally solvable, in one round ($m=1$)
in the wait-free model.
{Then, {exploiting} a known reduction between renaming and weak symmetry breaking,
we derive locally solvable valency tasks for renaming.}
This theorem implies our main result that there are no local proofs,
in the style of FLP, for set agreement and renaming,
as shown in Section~\ref{sec:locValImpProof},
where we present a precise notion of \emph{local impossibility proof}.
The techniques are based on  combinatorial topology arguments
explaining how a protocol can ``hide'' the inevitable mistakes it
must make in some final decisions.

The setting used is a round-based wait-free model, where
$n$ asynchronous processes communicate reading and writing shared variables.
Since the model is wait-free,
the impossibility results are related to $k$-set agreement,
$k=n-1$,  and $M$-renaming, $M=2n-1$.
Working in a round-based model facilitates identification of consistent
layers of configurations, and talking about $\ell$-round configurations.
Considering wait-free executions allows to assume the hypothetical protocol
decides always after some number of rounds, $R$.
The significance of these specific cases and the choice of the model
is further discussed in Section~\ref{section:them},
which also explain the relation of our results to the approach of
Alistarh, Aspnes, Ellen, Gelashvili and Zhu~\cite{AlistarhAEGZ2019},
the first paper that has considered this question,
which showed that \emph{extension-based} techniques do not suffice
for proving the impossibility of solving set agreement.

%% file: model.tex

\section{Model of Computation and Its Topological Interpretation}
\label{sec:model}

The model we  consider is a standard shared-memory system
with $n \geq 2$ asynchronous wait-free processes,
$P_0, \hdots, P_{n-1}$, communicating by atomically
reading and writing to shared variables.

\paragraph{The IIS model}
A \emph{protocol} specifies, for each process,
the steps to perform in order to solve a task.
We consider an \emph{iterated immediate
snapshot} (IIS)~\cite{iteratedRajsbaum10}
model of computation in which the protocol proceeds in
a sequence of asynchronous \emph{rounds}.
In each round $r \geq 1$, a process performs an
\emph{immediate snapshot} (IS) operation
on a clean shared array $M[r]$.
The execution of an IS operation on $M[r]$ is described as a sequence
of \emph{concurrency classes},
i.e., non-empty sets of processes.
Each concurrency class indicates that the processes in the class first
write in $M[r]$ (in some arbitrary order) and then read all entries of $M[r]$
(in some arbitrary order).
Each process appears in exactly one concurrency class
for round $r$, namely, executes one IS, on each memory $M[r]$.

An \emph{execution} starting in $\sigma$ is defined by a sequence of
IS executions, one for each $M[r]$:
the sequence of concurrency classes on $M[1]$, followed by the
sequence of concurrency classes on $M[2]$,
and so on.
Since processes access a clean memory $M[r]$ in every round $r$,
IIS executions can be equivalently defined as a sequence of
concurrency classes with the property that,
for each concurrency class $C$,
the processes in it perform the same number of IS operations
in the concurrency classes preceding $C$.
This means that all of them are poised to perform an IS operation
on the same $M[r]$.

A \emph{configuration} of the protocol
$\sigma=\{ (P_0, v_0),\ldots,(P_{n-1},v_{n-1})\}$
consists of the local state  $v_i$ for each process $P_i$,
during an execution.
Notice that the states of the processes define the
values assigned to the entries of $M[r]$.
In an \emph{initial configuration} $\sigma$,
each process of $\sigma$ is in an
initial state determined by its input  value (and its id),
and all shared variables hold their initial value.
A \emph{partial} configuration of a configuration $\sigma$
is a subset of $\sigma$.

\paragraph{Tasks}
A \emph{task} $T = (I, O, \Delta)$ is specified by a set of input
assignments $I$ to the processes participating in an execution,
a set of possible output assignments $O$ to the participating processes,
and a mapping $\Delta : I \mapsto 2^O$ specifying
the allowable outputs for each input assignment.
A protocol \emph{solves a task $T$} if in every execution starting
in any initial configuration $\sigma\in I$,
every participating process of $\sigma$ decides an output value,
such that the output values
of the processes respect $\Delta$ for their input values.
The  \emph{safety} property is  that  the decisions of the processes
starting with inputs $\sigma \in I$ define an output simplex $\tau$,
such that $\tau\in \Delta(\sigma)$.
The \emph{liveness} property is that the protocol is \emph{wait-free},
namely, a process does not take an infinite number of steps without deciding.

A task is solvable in the IIS model
if and only if is solvable in the  standard asynchronous read/write
model~\cite{BorowskyG1993,GafniR2010}.
When one is  interested only in computability (and not complexity),
the protocol may be assumed to be \emph{full-information}:
a process remembers everything, and always writes all the information it has.
Therefore, the protocol only needs to instruct a process
when to \emph{decide}, and on which output value.

The following  tasks are defined over a domain of possible inputs
$V= \{0,1,\ldots,n-1\}$.
For proving impossibility results, it suffices to assume that
a process $P_i$ starts with input $i$.

\begin{definition}
In the \emph{$k$-set agreement} task~\cite{Cha93}
processes decide on at most $k$ different values,
among the input values they have observed.
The case where $k = 1$ and $V = \{0,1\}$,
is the  \emph{binary consensus} task.
\end{definition}

\begin{definition}
In the \emph{$M$-renaming} task~\cite{ABDPR90} processes start with
distinct values from a large domain,
and decide on distinct values from a smaller domain $\{0, \ldots, M-1\}$.\\
In the \emph{weak symmetry breaking} task~\cite{GRH06} processes
decide values  in  $\{0,1\}$, such that
not all of them decide the same value.
\end{definition}

If there is a protocol solving $(2n-2)$-renaming then there is
a  protocol solving weak symmetry breaking~\cite{GRH06}.
Due to its simpler structure and equivalence to $(2n-2)$-renaming,
we study weak symmetry breaking instead of  studying renaming.

\paragraph{Topological Interpretation}
Since  protocols preserve topological invariants of
the model of computation,
and these invariants, in turn, determine which
tasks are solvable, it is convenient to describe
protocols in the topological model of distributed
computing~\cite{HerlihyKR2013}.

In this model, the inputs of a task form an \emph{input complex} $I$,
which is a family of sets closed under containment.
Each set in the family is called a \emph{simplex}.
An input simplex $\sigma\in I$ has the form $\sigma=\{ (P_i, v_i)\}$,
for some subset of processes $P_i$, denoted $ids(\sigma)$.
It indicates that process $P_i\in ids(\sigma)$ starts with input $v_i$.
The values  $v_i$ are taken from a universe $V$ of possible input values.
The \emph{facets} of $I$ are the simplexes of size $n$,
defining the initial configurations of the system.
(A \emph{facet} is a simplex that is not contained in another simplex.)
The output complex $O$ is defined similarly.

For each input simplex $\sigma \in I$,
a \emph{task} $T = (I, O, \Delta)$ specifies an output simplex
$\tau\in \Delta(\sigma)$, $\tau=\{ (P_i,v'_i)\}$.
This means that $P_i$ may decide $v'_i$,
in an execution starting with inputs defined by $\sigma$,
where the processes observe steps by processes in $ids(\sigma)$.

Consider tuples of the form $(P, view)$,
where $P$ is in $ids(\sigma)$ and $view$ is the state of $P$
after $\ell$ rounds of communication.
A configuration is a simplex, a set of such tuples,
specifying the states of the processes after $\ell$ rounds.
The set of all configurations starting in $\sigma$,
after some number of rounds $\ell$ (including the partial configurations),
defines the \emph{protocol complex} $\chi^\ell(\sigma)$.
The configurations of $\chi^\ell(\sigma)$ are the simplexes
of this complex.
For a partial configuration $\sigma'\subset \sigma$,
$\chi^\ell(\sigma')$ is the subset of $\chi^\ell(\sigma)$
corresponding to executions where the processes
of $ids(\sigma')$ see only immediate snapshots by themselves.

In our model, the topological invariant preserved is
that a full-information protocol subdivides the input complex.

The protocol complex is denoted $\chi^\ell(\sigma)$, since it turns out
that it is the $\ell$-th \emph{chromatic subdivision} of $\sigma$.
For example, when $n=3$, a configuration may be drawn as a triangle,
as seen in Figure~\ref{figure:examples}(left).
The figure depicts the subdivision obtained after one round,
$\chi(\sigma)$, for three processes ($p=$ black, $q=$ grey, $r=$ white),
starting in one input simplex~$\sigma$.
It describes the sequences of concurrency classes that led to four of its simplexes.
Notice that a partial configuration, $\sigma'\subset \sigma$, $|\sigma|=3$,
is depicted as a vertex (state of one process) or as an edge (state of two processes),
contained in the triangle~$\sigma$.
The subdivision $\chi^2(\sigma)$ is obtained by replacing each
triangle $\tau$ of $\chi(\sigma)$,  by $\chi(\tau)$, and so forth.

A task $T$ is \emph{solvable} in $\ell$ rounds if and only if there is
a \emph{simplicial map} $\delta$ from the $\ell$-th chromatic subdivision
$\chi^\ell(I)$ to $O$ that respects $\Delta$, i.e.,
for every $\sigma \in I$,
$\delta(\chi^\ell(\sigma))$ is a subcomplex of $\Delta(\sigma)$.
(A simplicial map sends vertices of one complex to vertices
of another complex, preserving simplexes.)

If the input complex is finite
(i.e., the universe $V$ of possible input values is finite),
it is well-known that there is an integer $R$,
such that processes always decide at
the end of the $R$-th round in a wait-free protocol.
(This follows directly from K\"onig's~Lemma.)

The dimension of the protocol complex,
as well as the input complex, is $n-1$.
(The \emph{dimension of a simplex} $\sigma$ is $|\sigma|-1$,
and the \emph{dimension of a complex} is the largest dimension
of any of its simplexes.)

The \emph{carrier}, $carr(\tau,\chi^{\ell}(\sigma))$,
is the smallest $\sigma'\subseteq \sigma$,
such that $\tau\in \chi^{\ell}(\sigma')$.
In the figure, for the two edges of $\tau$,
we have $carr(\tau',\chi^{\ell}(\sigma))=\sigma'$,
and $carr(\tau'',\chi^{\ell}(\sigma))=\sigma$.

A \emph{carrier map} $\Delta : I \mapsto 2^O$
sending each input simplex $\sigma \in I$
to a subcomplex $\Delta(\sigma)$ of $O$,
such that $\sigma\subseteq \sigma'$ implies
$\Delta(\sigma)\subseteq \Delta(\sigma')$. 

%% file: local-solvability.tex

\section{Valency Tasks and Local Solvability}
\label{sec:valTaskLocSolv}

We introduce here the notions of {valency task}, and of locally solving such a task.
Together, these notions provide the basic step in an impossibility proof in the FLP style,
that will be formally defined in Section~\ref{sec:locValImpProof}.

As discussed above, both for set agreement and weak symmetry breaking,
one may consider, without loss of generality,
a single input configuration,  $\sigma=\{ (P_0,0),\cdots, (P_{n-1},n-1)\}$,
meaning that the initial local states of the processes differ only in their ids.
Thus,
the input complex $I$ consists of $\sigma$ together with each subset of $\sigma$.
For short, let $\sigma=\{ 0,\cdots, {n-1}\}$, and we
sometimes abuse notation and denote the input complex also by $\sigma$.

Now, assume by way of contradiction that there is a protocol $\cal P$
solving an unsolvable task $\cal T$ in  $R$ rounds, for some $R\geq 1$. Namely, the protocol complex is  $\chi^{R}(\sigma)$,
and each vertex $v=(p,view)$ of this complex corresponds to the state $view$ of a process $p$, based on which, $p$  produces an output,
after executing an IS on $M[R]$.
Solving the task means that the protocol determines a simplicial map $\delta$, a coloring of each vertex $v$ of  $\chi^{R}(\sigma)$
with a decision value, $\delta(v)$, by the map $\delta(v)=(p,out)$, in such a way that for any final configuration $\tau\in\chi^{R}(\sigma)$, the simplex of
decision values $\delta(\tau)$ belongs to $\Delta(\sigma)$.
Since the task is unsolvable, there is no such $\delta$.
Intuitively, a local proof  demonstrates a contradiction by pinpointing
a configuration $\tau$ of the protocol complex where the decisions do
not satisfy the task specification, through a local observation, as follows.

\subsection{Overview of the local solvability approach}
\label{sec-local-valency-proofs}

\begin{figure}
\begin{center}
\includegraphics[width=0.45\textwidth]{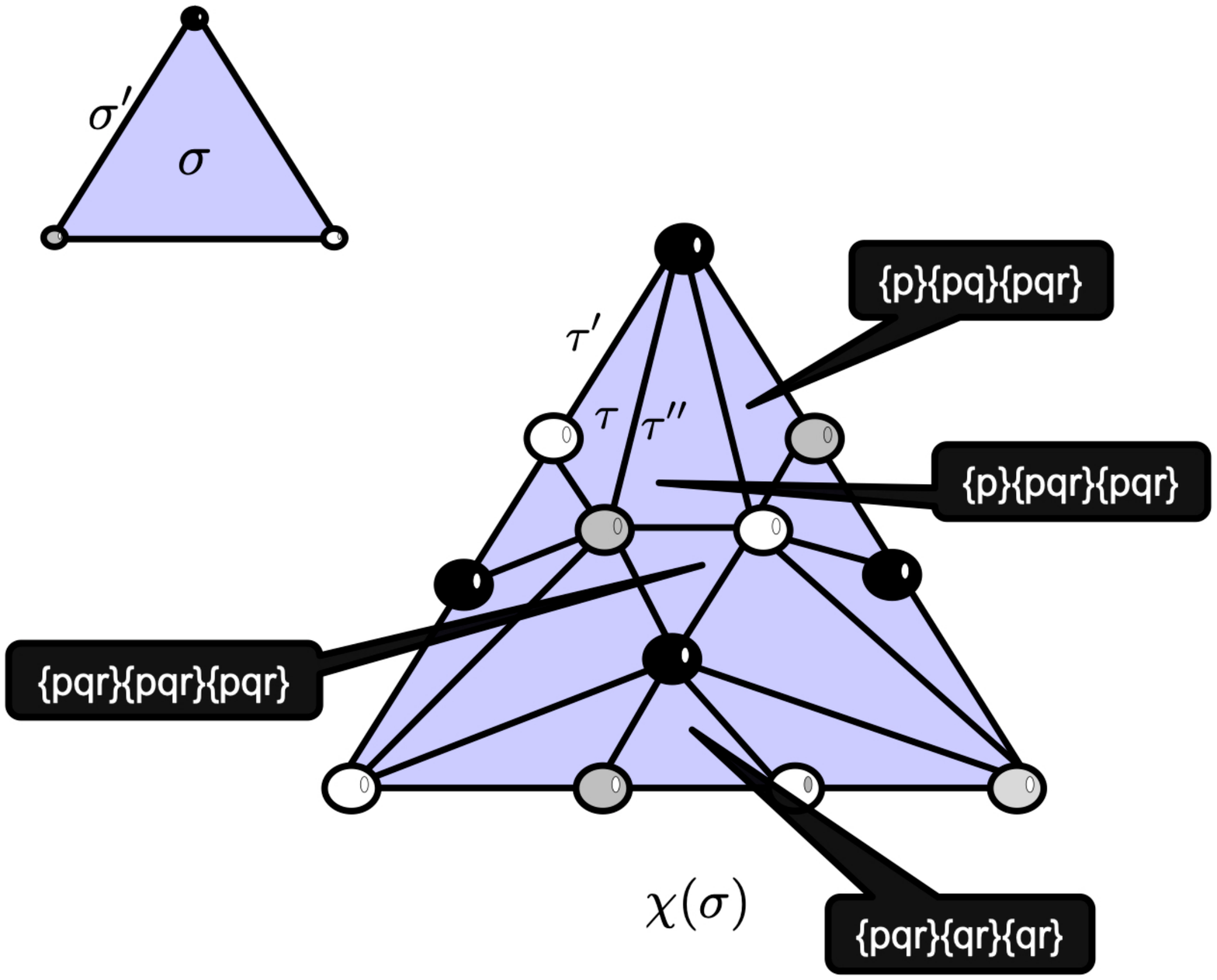}
\includegraphics[width=0.45\textwidth]{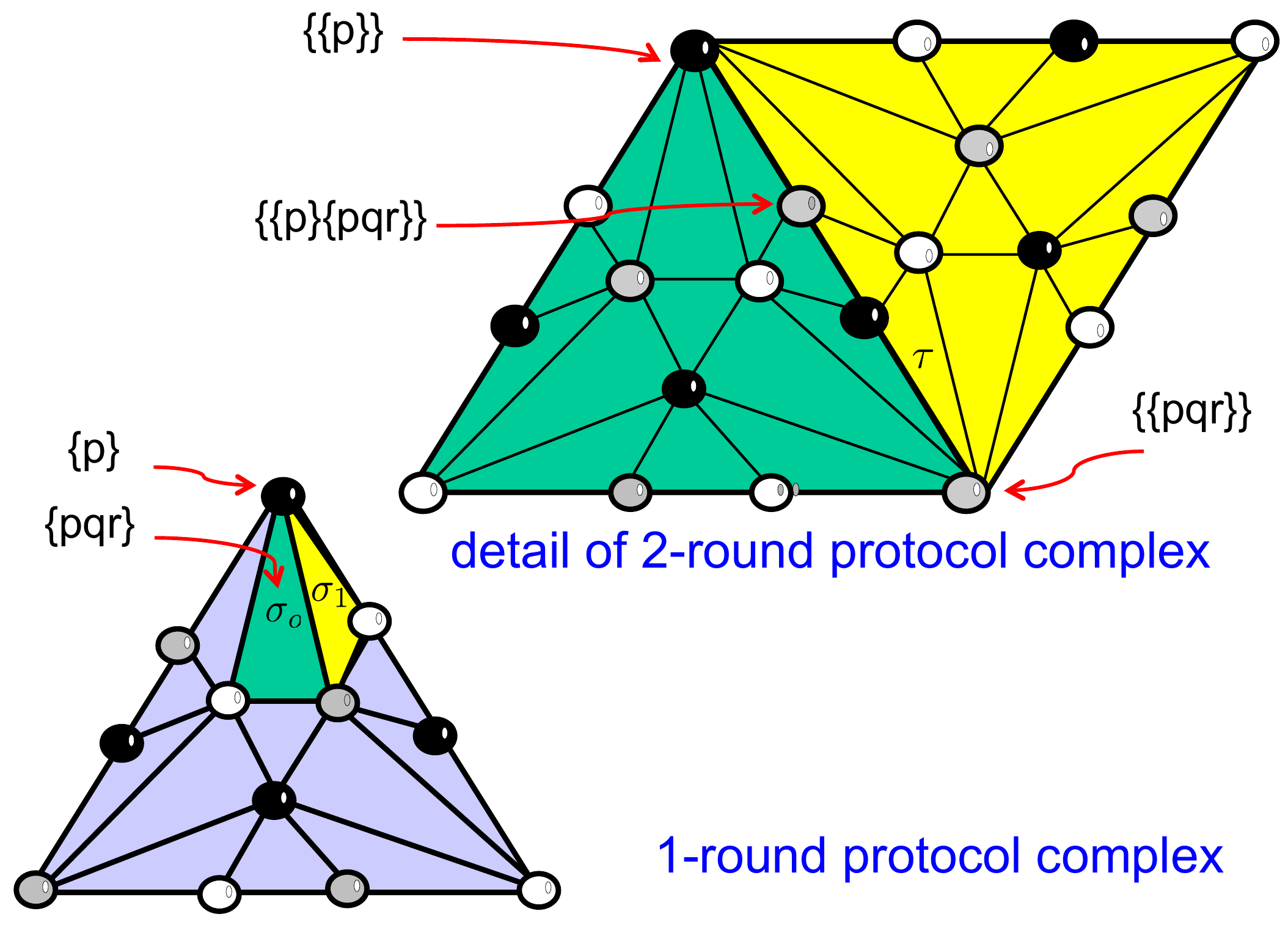}
\caption{Examples.}\label{figure:examples}
\end{center}
\end{figure}

Assume a protocol $\cal P$ solving the task in  $R=\ell+m$ rounds,
and
consider all the configurations after $\ell\geq 1$ rounds, $\chi^{\ell}(\sigma)$, and for each
configuration $\sigma'\in\chi^{\ell}(\sigma)$, the  \emph{valencies}, $val(\sigma')$ determined by $\cal P$.
Namely, for each value $v\in \textit{val}(\sigma')$,
there is a final configuration $\tau\in \chi^m(\sigma')$,
a successor of $\sigma'$ after $m$ rounds, such that
at least one process decides on the value $v$ in $\tau$.
The successor configurations of $\sigma'$ are all configurations after
$m$ additional rounds of computation
by processes in $\sigma'$,
namely, all simplexes in $\chi^m(\sigma')$.
Figure~\ref{figure:examples} (right) depicts the case of $\ell=m=1$.
The successor $\tau$ of $\sigma'$  is reached from the initial
configuration $\sigma$ in $\ell+m$ rounds.
Given $\chi^{\ell}(\sigma)$ and all the valencies of all these configurations,
the impossibility argument consists of selecting one $\sigma'\in \chi^{\ell}(\sigma)$. If there
are legal decisions $\delta_{\sigma'}$ for all final configurations extending  $\sigma'$,
then the impossibility argument did not succeed in finding a contradiction, because $\delta_{\sigma'}$ could be the map used by $\cal P$.
This is precisely what we show for set agreement and weak symmetry breaking: one can define valencies, such that
for any such $\sigma'$ there is a protocol  $\delta_{\sigma'}$ solving the task locally at $\sigma'$.
The protocol  $\delta_{\sigma'}$ colors the vertices
of $\chi^m(\sigma')$ after executing $m$ rounds starting in $\sigma'$
and satisfies  the task specification, and additionally,
the a priori  made commitments expressed
by $val(\sigma')$ (each $val\in val(\sigma')$ is indeed decided,
i.e., there is a vertex $(p,view)\in\chi^m(\sigma')$,
with  $\delta_{\sigma'}(p,view)=(p,val)$).
Thus, the protocol indeed preserves the valencies.

That is, an incorrect protocol can always hide the error locally.
Given that the task is unsolvable, an error must exist somewhere.
However, each particular configuration $\sigma'$ inspected looks fine,
and the error is moved elsewhere.
We stress that this holds for every $\ell \geq 1$ and $m=1$,
namely, even inspecting one round before the protocol terminates.

\subsection{There is no Locally Solvable Valency Task for Consensus}
\label{sec:consImp}

For consensus,
there is no way of defining a locally-solvable valency task.
This is indeed what is expected,
since there is a local impossibility proof for consensus.
We show that there is no way to assign valencies,
so that a protocol can hide its error.
We present the case where the hypothetical protocol solves consensus
in two rounds, $\ell+m=2$, but the general case is analogous.
(See Figure~\ref{figure:flp}.)
\begin{figure}
\begin{center}
\includegraphics[width=0.4\textwidth]{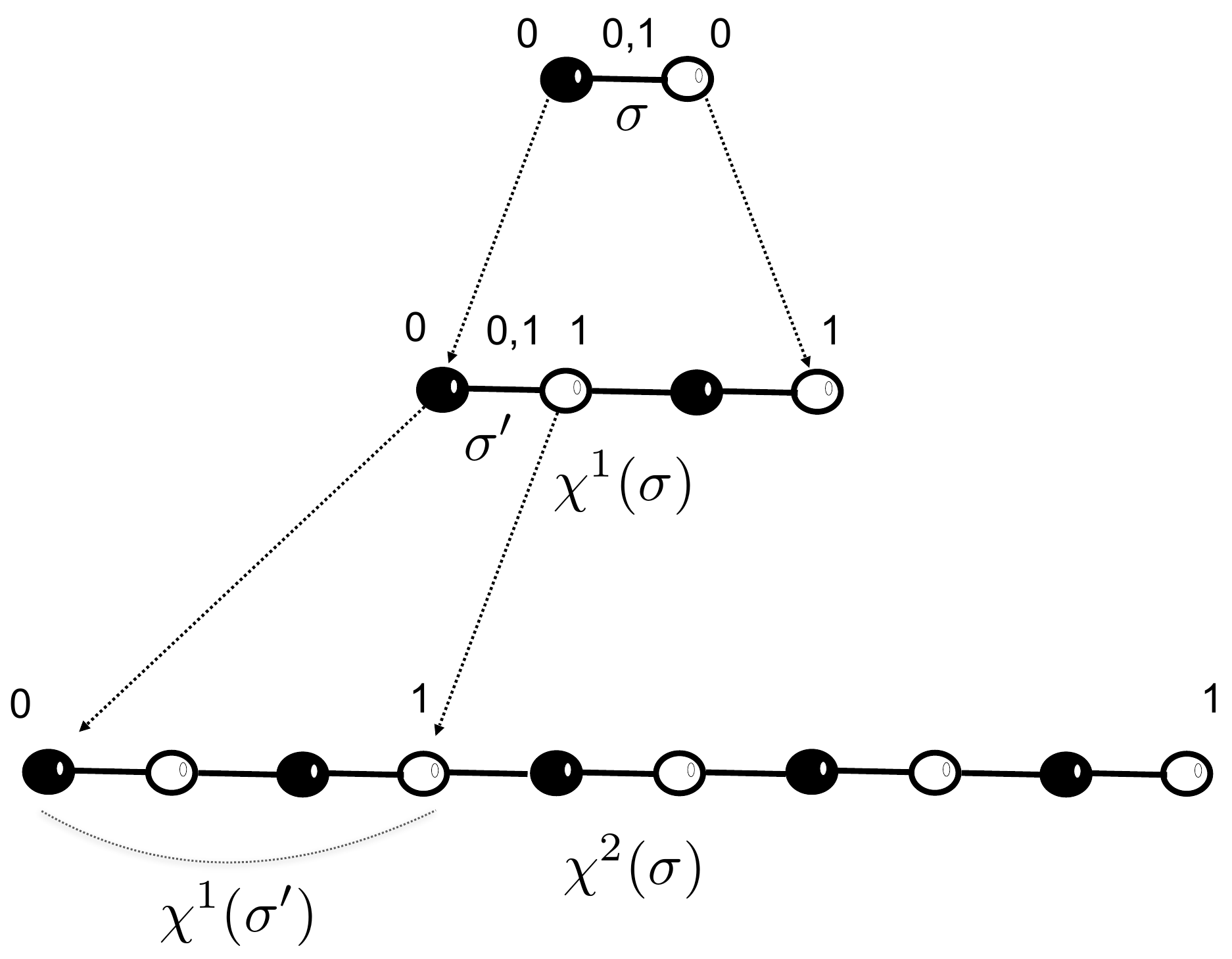}
\caption{Consensus is not locally solvable.}\label{figure:flp}
\end{center}
\end{figure}

Let $\sigma= \{ 0,1 \}$ be the input edge, and the task specification
$\Delta(\{0\})=\{(P_0,0)\}$, $\Delta(\{ 1\})=\{(P_1,1)\}$,
$\Delta(\{ 0,1\})=\{ \{ (P_0,0),(P_1,0)\},\{ (P_0,0),(P_1,1)\} \}$.
In terms of valencies,
for $i \in \{ 0,1 \}$, observe that $val(\{ i\}) = \{ i \}$ (for any $m$), because
$\chi^m(P_i, i)$ is the solo execution of $P_i$ with input $i$,
in which $P_i$ must decide $i$.
Thus, $val(\sigma) = \{0,1\}$ as $\chi^m(P_i, i) \subset \chi^\ell(\sigma)$.

Consider the complex $\chi^1(\sigma)$, which has the following edges:
$\{ (P_0, \langle 0 \rangle) , (P_1, \langle 0,1 \rangle) \}$,
corresponding to the execution in which $P_0$ goes first and then $P_1$;
 $\{ (P_1, \langle 0,1 \rangle) , (P_0, \langle 0,1 \rangle) \}$,
corresponding to the execution in which both processes run concurrently;
$\{ (P_0, \langle 0,1 \rangle) , (P_1, \langle 1 \rangle) \}$,
corresponding to the execution in which $P_1$ goes first and then $P_0$.
As explained above, $val(P_0, \langle 0 \rangle) = \{0\}$ and $val(P_1, \langle 1 \rangle) = \{1\}$, and
the valency of any other vertex of $\chi(\sigma)$,
$(P_0, \langle 0,1 \rangle)$ and $(P_1, \langle 0,1 \rangle)$,
is either $\{0\}$ or  $\{1\}$.
Thus, there must be an edge $\sigma' = (u,v) \in \chi^1(\sigma)$ among
the three edges with $val(u) = \{0\}$ and $val(v) = \{1\}$,
and hence $val(\sigma') = \{0,1\}$.
We pick such an edge $\sigma'$ and observe that consensus is not locally
solvable in  $ \chi^1(\sigma')$, i.e., the valency task with input $\sigma'$ and outputs $\chi^1(\sigma')$
with these valencies is not solvable. This is because any attempt to color the vertices of  $ \chi^1(\sigma')$,
with one endpoint of the path colored 0 and the other colored 1,
will produce an  edge $\tau$ whose vertices have different colors, violating the
agreement requirement of consensus.

We have seen that for consensus (1-set agreement) it is  impossible
to define a valency task that is locally solvable.
In Sections~\ref{sec:valTaskSA} and~\ref{sec:valTaskSB} we show how to
specify valency tasks for set agreement
and weak symmetry breaking that \emph{are} locally solvable,
and in Section~\ref{sec-almost-global} we describe protocols that solve them.

\subsection{Valency Tasks and Local Solvability for Set Agreement}
\label{sec:valTaskSA}

Consider now  the unique input simplex
$\sigma = \{0, \hdots, n-1\}$ for  $k$-set agreement, $k=n-1$.
Processes decide values from $\sigma$ that they have seen,
and such that at most $n-1$ different values are decided in an execution.

Following topology terminology, in the rest of the paper configurations are called simplexes.
First, recall that for a simplex $\tau \in \chi^\ell(\sigma)$,
the \emph{carrier}
of $\tau$ in  $\chi^\ell(\sigma)$ is  the smallest face $\sigma'\subseteq \sigma$,  such that
$\tau \in \chi^\ell(\sigma')$.
From an operational perspective, $carr(\tau, \chi^\ell(\sigma))$
identifies the set of  processes seen in the $\ell$-round
IIS execution that ends at configuration $\tau$.

The goal is to define, for each $\ell$,  a \emph{set agreement valency task} ${\cal T} = \langle \chi^{\ell}(\sigma), \sigma, val \rangle$.
This is a task that respects the set agreement specification:  a decided value should have been seen,
namely, a process deciding $v$ must have $v$ in its view.
Indeed, agreement tasks such as consensus and set agreement are specializations
of a  \emph{validity} task~\cite{CastanedaRR18jacm},  where this is the only requirement.

More formally, in a valency task for set agreement,  for every simplex $\tau \in \chi^{\ell}(\sigma)$,
$
val(\tau) \subseteq carr(\tau, \chi^{\ell}(\sigma)).
$
The set of inputs of $\cal T$ are the configurations at round $\ell$, namely $\chi^\ell(\sigma)$.
For each configuration $\tau \in \chi^\ell(\sigma)$,
the set of possible decisions $val(\tau)$ is a non-empty subset of~$\sigma$ (this is the standard
hypothesis of Sperner's lemma).
Notice that \textit{val} can be formally defined as a carrier map.\footnote{Formally,
the corresponding task specification $\Delta$,  for $\Delta(\tau)$, consists of all output simplexes labeled by output values
from $val(\tau)$.}
The following is a particular set agreement valency task.

\begin{definition}[Locally solvable set agreement valency task]
\label{def:locSolvSAtask}
For every integer $\ell \geq  1$, let ${\cal T}=\langle \chi^{\ell}(\sigma), \sigma, val \rangle$, where \textit{val} is
the carrier map defined by
\begin{enumerate}
\item If $|\tau| \leq n-2$, then $val(\tau) = ids(\tau)$,
\item else $val(\tau) = carr(\tau, \chi^{\ell}(\sigma))$.
\end{enumerate}
\end{definition}

In the notion of \emph{local solvability} of valency-tasks, we ask for a protocol that solves $\cal T$ in $m$ rounds,
namely a decision map $c_{\tau}: \chi^{\ell+m}(\sigma) \rightarrow \sigma$ that respects $\cal T$.
Thus, $c_{\tau}$ is a \emph{global} solution to $\cal T$, but
the $k$-set agreement task is solved only \emph{locally} at $\tau$: $c_{\tau}$ is determined by
 a specific input simplex $\tau \in \chi^\ell(\sigma)$, and
$c_{\tau}(\chi^m({\tau}))$ does not have any simplex with $k+1$ decisions.
Of course,  $c_{\tau}$ \emph{does not globally} solve $k$-set agreement because
indeed $c_{\tau}(\chi^{\ell+m}(\sigma))$ is a Sperner's coloring and has  at least one simplex colored with $n$ different decisions,
 by Sperner's lemma~\cite{spernerL}.
Recall that
a Sperner coloring $c: \chi^{\ell+m}(\sigma) \rightarrow \sigma$ is a simplicial map such that
$c(v) \in carr(\tau, \chi^{\ell+m}(\sigma))$, for every vertex $v$ of $\chi^{\ell+m}(\sigma)$.

We have that if $c_{\tau}$ solves $\cal T$ in $m$ rounds,
for each
$\tau'\in\chi^\ell(\sigma)$
and all configurations after $m$ rounds, $\chi^{m}(\tau')$,
it should hold that $c_\tau$ is \emph{consistent}, i.e., $c_\tau(\chi^{m}(\tau'))\subseteq val(\tau')$.
We require that $c_\tau$ is additionally
 \emph{complete}, meaning that every value committed by the valencies, is indeed
 decided, namely, $c_\tau(\chi^{m}(\tau'))= val(\tau')$.

\begin{definition}[Local solvability of $k$-set agreement]
We say that a set agreement valency task ${\cal T}=\langle \chi^{\ell}(\sigma), \sigma, val \rangle$
is $k$-\emph{locally solvable}
in $m \geq 1$ rounds if for every input simplex
$\tau\in \chi^\ell(\sigma)$ there is a decision simplicial map $c_\tau: \chi^{\ell+m}(\sigma) \rightarrow \sigma$
that is consistent and complete w.r.t. $\cal T$ and $c_\tau(\chi^m(\tau))$ does not have simplexes with
more than $k$ distinct decisions at its vertices.
\end{definition}

We stress that local solvability allows $c_\tau$ (which depends on $\tau$)
to have simplexes not in $\chi^m(\tau)$ with more than $k$ distinct decisions,
as it requires that $c_\tau$  solves $k$-set agreement only in $\chi^m(\tau)$.
Although it is unavoidable that there are simplexes with more than $k$ distinct
decisions \emph{somewhere} (due to the $k$-set agreement impossibility),
local solvability does not require that the task is globally unsolvable.
Indeed, while we prove (Section~\ref{sec:SAlocalSolving}) that
the valency task for set agreement is $(n-1)$-locally solvable in a single round,
we do not prove it is globally unsolvable.

\subsection{Valency Tasks and Local Solvability for Weak Symmetry Breaking {and $(2n-2)$-Renaming}}
\label{sec:valTaskSB}

The weak symmetry breaking {and $(2n-2)$-renaming tasks} with unique input $(n-1)$-simplex
$\sigma = \{0, \hdots, n-1\}$
require that output colorings on the boundary of
$\chi^\ell(\sigma)$
have the next symmetry property
(assuming the protocol terminates in $\ell$ rounds)
on the vertices $V(\chi^{\ell}(\sigma))$, e.g.~\cite{CR10, CR12, HS99}:

\begin{definition}[Symmetric {output} coloring]
\label{def-symmetric-outputs}
A \emph{symmetric {output} coloring} of $\chi^{\ell}(\sigma)$
is a simplicial map $b: V(\chi^{\ell}(\sigma)) \rightarrow {{\sf OutputSpace}}$ satisfying that,
for any two distinct proper faces $\sigma', \sigma''$ of $\sigma$
of the same dimension,
$v \in V(\chi^\ell(\sigma'))$ and $\phi(v) \in V(\chi^\ell(\sigma''))$
have the same {output} color, i.e. $b(v) = b(\phi(v))$,
where $\phi$ is the simplicial bijection between
$V(\chi^\ell(\sigma'))$ and $V(\chi^\ell(\sigma''))$ that maps
vertices \emph{preserving order}, namely, vertices with the smallest id in
$ids(\sigma')$ to vertices with the smallest id in $ids(\sigma'')$,
vertices with the second smallest id in $ids(\sigma')$ to vertices
with the second smallest id in $ids(\sigma'')$, and so on.
\end{definition}

{Recall that {any weak symmetry breaking or renaming} protocol can be
transformed into a \emph{comparison-based} protocol,
in which processes only perform comparisons between inputs~\cite{Herlihy2008}.}
Thus, actual input values are irrelevant,
and only the relative order among them matters.
In inputless {weak symmetry breaking or renaming},
$i \in \sigma$ denotes the process with $i$-th input,
in ascending order.

Output decisions in weak symmetry breaking are binary, hence in valency tasks for weak symmetry breaking
the carrier map \textit{val} goes from $\chi^\ell(\sigma)$ to $\{0,1\}$, the complex with a single edge, and its vertices.
Since \textit{val} models the valencies of a hypothetical protocol for weak symmetry breaking,
the valencies must be symmetric on the boundary; this is the only requirement \textit{val} must satisfy.
The following is a particular weak symmetry breaking valency task, where it is not hard to check that  \textit{val}
is indeed a carrier map.

\begin{definition}[Locally solvable weak symmetry breaking valency task]
\label{def:locSolvWSB}
For every  $\ell \geq  1$, let  ${\cal T}= \langle \chi^{\ell}(\sigma), \{0,1\}, val \rangle$
where \textit{val} is the  carrier map defined by
\begin{enumerate}
\item If $dim(\tau) \leq n-3$, then $val(\tau) = \{1\}$.
\item Otherwise, $val(\tau) = \{0,1\}$.
\end{enumerate}
\end{definition}

Analogous to  set agreement, if a symmetric binary coloring
$b_\tau: V(\chi^{\ell+m}(\sigma)) \rightarrow \{0,1\}$ solves $\cal T$
in $m$ rounds then it respects \textit{val}, or is  \emph{consistent} with $\cal T$.
This means that for every input simplex $\tau'\in \chi^\ell(\sigma)$,
$b_\tau(\chi^{m}(\tau'))\subseteq val(\tau')$.
We also require that it is \emph{complete},
i.e., $b_\tau(\chi^{m}(\tau'))=val(\tau')$.

It has been shown~\cite{AttiyaP2016,CR10} that if $dim(\sigma)+1$ is a prime power, then
$b(\chi^{\ell}(\sigma))$ has at least one \emph{monochromatic} simplex
(i.e. with all its vertices having the same binary color) of dimension $dim(\sigma)$,
which implies the impossibility of weak symmetry breaking;
those monochromatic simplexes are the errors that $b$ makes,
however, $b$ is able to hide them locally:
for the specified input simplex $\tau \in \chi^\ell(\sigma)$,
$b(\chi^m({\tau}))$ does not have monochromatic simplexes of dimension $dim(\sigma)$.

\begin{definition}[Local solvability of weak symmetry breaking]
We say that a weak symmetry breaking valency task
${\cal T} = \langle \chi^{\ell}(\sigma), \{0,1\}, val \rangle$  is \emph{locally solvable}
in $m \geq 1$ rounds if for every input simplex
$\tau\in \chi^\ell(\sigma)$ there is a symmetric binary decision map $b_\tau: \chi^{\ell+m}(\sigma) \rightarrow \{0,1\}$
(which is on function of $\tau$)
that is consistent and complete w.r.t. $\cal T$ and $b_\tau(\chi^m(\tau))$ does not have
monochromatic simplexes of dimension $dim(\sigma)$.
\end{definition}

In the next section, we prove that the weak symmetry breaking
valency tasks in Definition~\ref{def:locSolvWSB}
are locally solvable in one round.
This result is trivial when $dim(\sigma)+1$ is not a prime power because in those cases
weak symmetry breaking is indeed solvable~\cite{CR12}, and hence,
there is a symmetric binary coloring with no monochromatic simplexes
(i.e., without errors).
The interesting case in when weak symmetry breaking is not solvable
and unavoidable errors need to be hidden.

{Valency tasks and local solvability for $(2n-2)$-renaming are defined very similarly, 
{with} the only difference {being} that
valencies and outputs are taken from the set $\{1, 2, \hdots, 2n-2\}$ and local solvability requires that
no pair of vertices of the same simplex in $b_\tau(\chi^m(\tau))$ decide the same output name.
Later we derive $(2n-2)$-renaming locally solvable valency tasks from the weak symmetry breaking valency tasks in Definition~\ref{def:locSolvWSB}.}

%% file: almost-global.tex
\section{Solving Valency Tasks}
\label{sec-almost-global}

This section contains the proof of Theorems~\ref{theo-colorings-sa} and Theorem~\ref{theo-colorings-wsb},
stating that the set agreement and weak symmetry breaking valency tasks defined in the previous section,
Definitions~\ref{def:locSolvSAtask} and~\ref{def:locSolvWSB}, are locally solvable in one round.
{Using Theorem~\ref{theo-colorings-wsb}, it also proves Theorem~\ref{theo-colorings-ren} stating that there
exists $(2n-2)$-renaming valency tasks that are solvable.}

\subsection{Set Agreement}
\label{sec:SAlocalSolving}

The following theorem shows that the valency tasks for set agreement defined
in the previous section are locally solvable.
\begin{theorem}
\label{theo-colorings-sa}
For any $n\geq 3$ and $\ell \geq  1$,
the set agreement valency task ${\cal T}=\langle \chi^{\ell}(\sigma), \sigma, val \rangle$ in Definition~\ref{def:locSolvSAtask}
is $(n-1)$-locally solvable in one round.
\end{theorem}

The proof of Theorem~\ref{theo-colorings-sa} relies on the following lemma,
regarding vertex colorings of the first standard chromatic subdivision.
Roughly speaking, the lemma identifies colorings that, to some extent,
satisfy the properties of a Sperner coloring, but without simplexes with $n$ different decisions.
Figure~\ref{fig-example-sa} presents an example of these colorings.

\begin{figure}[t]
\begin{center}
\vspace{0.4cm}
\includegraphics[scale=0.4]{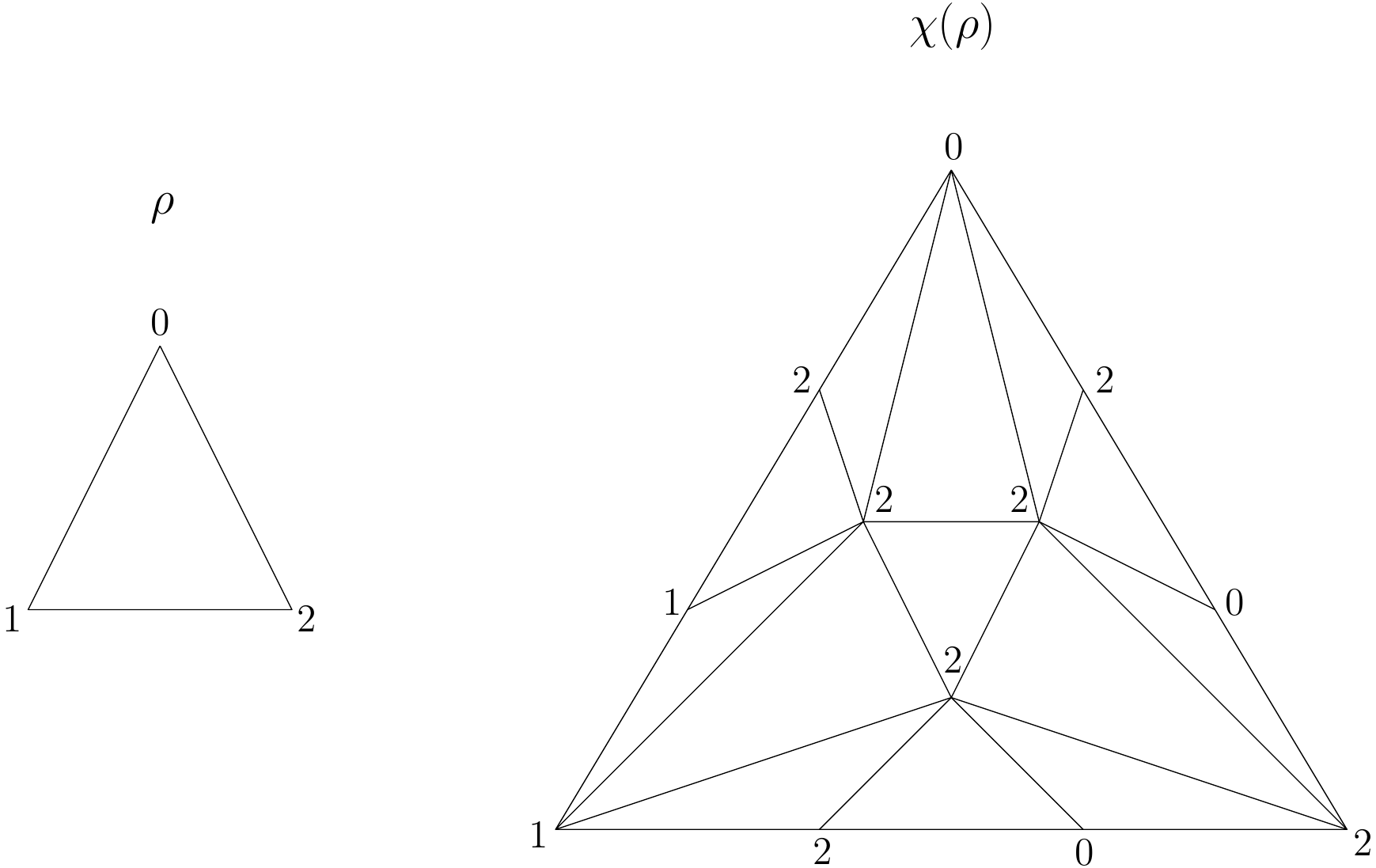}
\caption{\footnotesize Example for Lemma~\ref{lemma-color-sa}. 
Since every vertex $v$ of $\rho$ has dimension $n-3$, the only vertex in $\chi(v)$ has color $v$, implying (1).
The coloring satisfies requirement (2.b) of the lemma: for the $(n-2)$-face $\{0,2\}$ of $\rho$,
all vertices in $\chi(\{0,2\})$ have colors in $\{0,2\}$, while for any other $(n-2)$-face $\rho'$ of $\rho$,
every vertex in $\chi(\rho')$ has a color in $\rho$. Finally, $\chi(\rho)$ has no $(n-1)$-simplex with the three colors at its vertices,
implying (3).
}
\label{fig-example-sa}
\end{center}
\end{figure}

\begin{lemma}
\label{lemma-color-sa}
Consider the $(n-1)$-dimensional simplex $\rho = \{0, \hdots, n-1\}$ with $n \geq 3$.
There is a coloring (simplicial map) $c: \chi(\rho) \rightarrow \rho$ such that:
\begin{enumerate}
\item for every $\rho' \subset \rho$ with $dim(\rho') \leq n-3$, $c(\chi(\rho')) = \rho'$,
\item one of the following holds:
\begin{enumerate}
\item for every $(n-2)$-dimensional face $\rho' \subseteq \rho$, $c(\chi(\rho')) = \rho$,
\item for a chosen $(n-2)$-dimensional $\rho'' \subset \rho$, $c(\chi(\rho'')) = \rho''$,
and for every other $(n-2)$-dimensional face $\rho' \subseteq \rho$, $c(\chi(\rho')) = \rho$,
\end{enumerate}
\item $c(\chi(\rho)) = \rho$ and there is no fully colored $(n-1)$-simplex in $c(\chi(\rho))$.
\end{enumerate}
\end{lemma}

\begin{proof}
The proof of the lemma is based on the following claim.
In what follows, let $\rho_i$ denote the
$(n-2)$-face of $\rho$ without vertex $i \in \rho$.

\begin{claim}
\label{claim-SA}
For every $i, j \in \rho$,
there is a simplicial map $c: \chi(\rho_i) \rightarrow \rho$ such that:
\begin{enumerate}
\item for every face $\rho' \subset \rho$, for every vertex $v \in \chi(\rho_i)$, $c(v) = id(v)$,
\item there is a vertex $v \in \chi(\rho_i)$ with $c(v) = i$,
\item there is no $(n-2)$-simplex $\gamma \in \chi(\rho_i)$ with $c(\gamma) = \rho_j$.
\end{enumerate}
\end{claim}

\begin{proof}
Observe that the coloring of the boundary of $\chi(\rho_i)$ is already defined,
thus we just need to define the coloring for the vertices of the central $(n-2)$-simplex simplex
of $\chi(\rho_i)$, namely, the (unique) $(n-2)$-simplex $\rho$ such that each $v \in \rho$
has carrier $\rho_i$. Notice that at least one of the vertices of $\rho$ must have color $i$.

Consider first the case $i = j$.
For every $v \in \rho$, set $c(v) = i$.
Thus, no $(n-2)$-simplex $\gamma \in \chi(\rho_i)$ can have $c(\gamma) = \rho_i$
because every $(n-2)$-simplex of $\chi(\rho_i)$ has at least one vertex of $\rho$.

Consider now the case $i \neq j$. Pick a $k \in \rho_i$ distinct to $j$.
For the vertex $v \in \rho$ with $id(v) = k$, set $c(v) = i$,
and for every other vertex $u \in \rho$, set $c(u) = j$.
The only way an $(n-2)$-simplex $\gamma \in \chi(\rho_i)$ can have $c(\gamma) = \rho_j$
is that $\gamma$ contains the vertex $v \in \rho$ with $c(v) = i$, whose $id$ is $k$.
Simplex $\gamma$ has the form $\gamma \cup \{ v \}$, for an $(n-3)$-simplex $\gamma \in \chi(\rho_{ik})$,
where $\rho_{ik}$ is the face of $\rho_i$ without $k$.
By construction, $c(\gamma) = \rho_{ik}$ and then $\gamma$ has a vertex with color $j$
as $k \neq j$. The claim follows.
\end{proof}

We use Claim~\ref{claim-SA} to prove the lemma.
First, set $c(v) = id(v)$ for $v \in \chi(\rho')$,
where $\rho' \subset \rho$ with $dim(\rho') \leq n-3$.
Thus, we have $c(\chi(\rho')) = \rho'$.
Also, set $c(v) = 0$,
for every vertex $v$ of the central $(n-1)$-simplex of $\chi(\rho)$,
namely, the (unique) $(n-1)$-simplex with all its vertices having carrier $\rho$.

Note that it remains to color the vertices of
the central $(n-2)$-simplex of $\chi(\rho_i)$,
for every $(n-2)$-face of $\rho$.
Since the central $(n-1)$-simplex of $\chi(\rho)$ is $0$-monochromatic,
the only way to have a fully colored $(n-1)$-simplex in $\chi(\rho)$
is that, for some face $\rho_i$ of $\rho$, there is an $(n-2)$-simplex $\gamma \in \chi(\rho_i)$
with $c(\gamma) = \rho_0$.
This simplex together with one vertex of the central $(n-1)$-simplex
are an $(n-1)$-simplex of $\chi(\rho)$, which would be fully colored.
We use Claim~\ref{claim-SA} to complete the coloring of
$\chi(\rho)$, avoiding such $(n-2)$-simplexes.

There are two cases to consider to prove in part (2) of the lemma.
In the first one, (2.a),
for every face $\rho_i \subseteq \rho$, we should have $c(\chi(\rho_i)) = \rho$.
Claim~\ref{claim-SA} implies that the coloring $c$ of $\chi(\rho_i)$
can be extended so that $c(\chi(\rho')) = \rho$ and
there is no $(n-2)$-simplex $\gamma \in \chi(\rho_i)$ with $c(\gamma) = \rho_0$.
As explained above, this guarantees that there are no fully colored $(n-1)$-simplexes in $\chi(\rho)$.

In the second case, (2.b),
for a chosen $\rho_j \subset \rho$, $c(\chi(\rho_j)) = \rho_j$,
and for every other face $\rho_i \subseteq \rho$, $c(\chi(\rho_i)) = \rho$.
Let us chose first $\rho_j = \rho_0$.
For every vertex $v$ of the central $(n-2)$-simplex of $\chi(\rho_0)$, set $c(v) = id(v)$.
For any other $\rho_i \neq \rho_0$, by the claim above, the coloring $c$ of $\chi(\rho_i)$
can be extended so that $c(\chi(\rho')) = \rho$ and
there is no $(n-2)$-simplex $\gamma \in \chi(\rho_i)$ with $c(\gamma) = \rho_0$.
Therefore, there is no fully colored $(n-1)$-simplex in $\chi(\rho)$, as explained before.

For the case we chose a face $\rho_j \neq \rho_0$, we simply permute the colors of the coloring just described
such that for the permuted coloring $c'$ we have $c'(\chi(\rho_0)) = \rho_j$.
The coloring $c'$ is induced by any permutation of $\rho$ that maps $\rho_0$ to $\rho_j$.
For example, in Figure~\ref{fig-example-sa} we can obtain a coloring for $\rho_j = \{0,1\}$ (instead of $\{0,2\}$)
by applying the permutation $0 \mapsto 0$, $1 \mapsto 2$ and $2 \mapsto 1$.
The lemma follows.
\end{proof}

For every simplex $\tau \in \chi^\ell(\sigma)$, let $ids(\tau)$ be the simplex containing
the first entries of the vertices in~$\tau$ (recall that each vertex of $\chi^\ell(\sigma)$ is a pair $(v, view)$
where $v$ is the id of a process and $view$ is its view after $\ell$ rounds);
note that $ids(\tau)$ is a $dim(\tau)$-face of $\sigma$, and $ids(\tau) \subseteq carr(\tau, \sigma, \chi^\ell)$,
since $\chi^\ell(\sigma)$ is a chromatic subdivision of $\sigma$.

\begin{proof}[Proof of Theorem~\ref{theo-colorings-sa}]
We prove now that set agreement valency task
${\cal T}=\langle \chi^{\ell}(\sigma), \sigma, val \rangle$
is $(n-1)$-locally solvable in one round.
We argue that $val$ is indeed the carrier map, and this  defines a
Sperner valency-task with input and output complexes $\chi^{\ell}(\sigma)$ and $\sigma$.
In the two rules above $val(\tau)$ is a non-empty face of $\sigma$.
Also, if $dim(\tau) \leq n-3$, then, $val(\tau) = ids(\tau)$,
and as already noted, $ids(\tau) \subseteq carr(\tau, \sigma, \chi^{\ell})$; otherwise
$val(\tau) = carr(\tau, \sigma, \chi^{\ell})$. In any case, $val(\tau) \subseteq carr(\tau, \sigma, \chi^{\ell})$.
It remains to be shown that $val$ is a carrier map.
Consider two simplexes $\tau', \tau \in \chi^{\ell}(\sigma)$ such that $\tau' \subseteq \tau$.
Observe that $ids(\tau') \subseteq ids(\tau)$ and $ids(\tau') \subseteq carr(\tau', \sigma, \chi^{\ell}) \subseteq carr(\tau, \sigma, \chi^{\ell})$.
If $dim(\tau') \leq n-3$, then $val(\tau') = ids(\tau')$ and $val(\tau)$ is either $ids(\tau)$ or $carr(\tau, \sigma, \chi^{\ell})$;
in both cases $val(\tau') \subseteq val(\tau)$.
And if $dim(\tau') > n-3$, then  $val(\tau') = carr(\tau', \sigma, \chi^{\ell})$, and also
$val(\tau) = carr(\tau, \sigma, \chi^{\ell})$ (as $\tau' \subseteq \tau$), and hence $val(\tau') = val(\tau)$.
Thus, we conclude that $val$ is carrier map.
To do so, we define a Sperner coloring $c_\tau$ of $\chi^{\ell+1}(\sigma)$
that is consistent and complete w.r.t.\ the task and
has no fully colored $(n-1)$-simplexes in $c_\tau(\chi(\tau))$,
for any input simplex $\tau \in \chi^\ell(\sigma)$.
We focus on the case  when $|\tau|=n$  because
for any simplex $\tau'$ of a smaller dimension, we can just pick any $\tau$ containing
$\tau'$, and set $c_{\tau'}$ to $c_\tau$ restricted to $\chi(\tau')$, i.e.~$c_\tau |_{\chi(\tau')}$.

Thus, for the rest of the proof fix an $(n-1)$-dimensional simplex of $\tau \in \chi^{\ell}(\sigma)$.
We define a Sperner coloring $c_\tau$
that is consistent with \textit{val} and has no fully colored $(n-1)$-simplexes in $\chi(\tau)$.
First, we use Lemma~\ref{lemma-color-sa} to define $c_\tau$ restricted to $\chi(\tau)$, i.e. $c_\tau |_{\chi(\tau)}$,
and then extend the coloring to all vertices in $\chi^{\ell+1}(\sigma)$, to finally obtain $c_\tau$.

Let $\rho = ids(\tau)$. Note that $\rho = \sigma$ but for clarity we use $\rho$.
ids's naturally induce a bijection between vertices of $\rho$ and $\tau$,
and $\chi(\rho)$ and $\chi(\tau)$, respectively,
hence any coloring (simplicial map) $\chi(\rho) \rightarrow \rho$ induces a coloring $\chi(\tau) \rightarrow ids(\tau)$.
Below, when we use Lemma~\ref{lemma-color-sa} applied to $\rho = ids(\tau)$, we can speak about faces of $\tau$ instead of faces of $\rho$.

Observe that either for every $(n-2)$-face $\tau'$ of $\tau$, $carr(\tau', \sigma, \chi^{\ell}) = \sigma$,
or for one $(n-2)$-face $\tau''$ of $\tau$, $carr(\tau', \chi^{\ell}) = ids(\tau'')$
and for every other $(n-2)$-face $\tau'$ of $\tau$, $carr(\tau', \sigma, \chi^{\ell}) = \sigma$.
Intuitively, $\tau'$ is ``inside'' $\chi^{\ell}(\sigma)$ or
only one $(n-2)$-face of $\tau'$ ``touches'' the boundary of $\chi^{\ell}(\sigma)$
(see Figure~\ref{fig-example-sa}).
We set $c_\tau |_{\chi(\tau)}$ using a coloring of $\chi(\tau)$ in Lemma~\ref{lemma-color-sa}, as follows.
In the former case,
$c_\tau |_{\chi(\tau)}$ is obtained with a coloring of $\chi(\tau)$,
as in Case~(2.a) of Lemma~\ref{lemma-color-sa},
while in the latter case,
is obtained with a coloring as in case (2.b), where $\tau''$ is the chosen face in that case of the lemma.

We argue that Lemma~\ref{lemma-color-sa} and the definition of \textit{val} implies that for any face $\tau'$ of $\tau$,
it holds that $c_\tau(\chi(\tau')) = val(\tau')$, which is good because we want $c_\tau$ to be consistent and complete w.r.t. \textit{val}.
If $dim(\tau') \leq n-3$, then $val(\tau') = ids(\tau')$, by definition of \textit{val},
and from Lemma~\ref{lemma-color-sa}(1), we know that $c_\tau(\chi(\tau')) = ids(\tau')$.
Also, by definition of \textit{val}, if $dim(\tau') > n-3$, then $val(\tau') = carr(\tau', \sigma, \chi^\ell)$.
Note that if $dim(\tau') = n-1$ (hence $\tau' = \tau$), then $val(\tau') = \sigma$,
and $c_\tau(\chi(\tau')) = val(\tau')$, by Lemma~\ref{lemma-color-sa}(3).
The subcase that remains to be shown is when $dim(\rho) = n-2$.
Again, if $val(\tau') = \sigma$, $c_\tau(\chi(\tau')) = val(\tau')$, by Lemma~\ref{lemma-color-sa}(3).
Thus, consider the case $val(\tau') = carr(\tau', \sigma, \chi^\ell) \neq \sigma$.
Observe that this can only happen when $\tau'$ is at the boundary of $\chi^\ell(\sigma)$,
and hence $carr(\tau', \sigma, \chi^\ell) = ids(\tau')$.
By Lemma~\ref{lemma-color-sa}(2.b), $c_\tau(\chi(\tau')) = ids(\tau')$
($\tau'$ was the chosen $(n-2)$-face of $\tau$ in case (2.b) of Lemma~\ref{lemma-color-sa} when defining $c_\tau$ on $\chi(\tau)$).

We now extend the coloring $c_\tau$ in two steps.
First, for any vertex $v \in \chi^\ell(\sigma)$ that does not belong to $\chi(\tau)$, we first set $c_\tau(v) = id(v)$.
Thus, for any input simplex $\lambda \in \chi^\ell(\sigma)$ that does not intersect $\tau$,
we have that $c_\tau(\chi(\lambda)) = ids(\lambda)$.
That is fine if $dim(\lambda) \neq n-2$, or $dim(\lambda) = n-2$ and $carr(\lambda, \sigma, \chi^\ell) \neq \sigma$,
because in such cases $c_\tau(\chi(\lambda)) = val(\lambda)$, by definition of \textit{val}.

But if $dim(\lambda) = n-2$ and $carr(\lambda, \sigma, \chi^\ell) = \sigma$,
then $c_\tau(\chi(\lambda)) = ids(\lambda) \subset val(\lambda) = \sigma$,
and then in this case $c_\tau$ is not complete.
Note that the proper contention is because there is no vertex in $\chi(\lambda)$ that is mapped to
the unique vertex in $\sigma \setminus ids(\lambda)$.
To solve this issue, for every such input simplex $\lambda \in \chi^\ell(\sigma)$,
we pick one vertex $v \in \chi(\lambda)$ with $carr(v, \lambda, \chi) = \lambda$
(which belongs to the ``central'' $(n-2)$-simplex of $\chi(\lambda)$)
and set $c_\tau(v)$ to the unique vertex in $\sigma \setminus ids(\lambda)$.
Therefore, we now have that $c_\tau(\chi(\lambda)) = val(\lambda) = \sigma$.

To fully prove that $c_\tau$ is consistent and complete w.r.t. \textit{val},
the only case that remains is of an input simplex $\lambda$ that
intersects $\tau$ but is not one of its faces.
Let $\lambda'$ and $\tau'$ be the proper faces of $\lambda$
and $\tau$ such that $\lambda = \lambda' \cup \tau'$.
We already know that $c_\tau(\chi(\lambda')) = val(\lambda')$
and $c_\tau(\chi(\tau')) = val(\tau')$.
If $dim(\lambda' \cup \tau') \leq n-3$, the definition of \textit{val} implies that
$val(\lambda \cup \tau') = val(\lambda') \cup val(\tau')$, hence
$c_\tau(\chi(\lambda' \cup \tau')) = val(\lambda \cup \tau')$.
If $dim(\lambda' \cup \tau') = n-1$, then it must be that $val(\lambda \cup \tau') = \sigma$,
from the definition of \textit{val}, and then clearly $c_\tau(\chi(\lambda' \cup \tau')) = val(\lambda \cup \tau')$,
by construction.
If $dim(\lambda' \cup \tau') \geq n-2$, we have two cases,
$val(\lambda' \cup \tau')$ is either $ids(\lambda' \cup \tau')$ or $\sigma$;
in any case, the very definition of $c_\tau$ implies that
$c_\tau(\chi(\lambda' \cup \tau')) = val(\lambda \cup \tau')$.

Therefore, so far we have a coloring $c_\tau$ that is consistent and complete w.r.t. \textit{val}
and $c_\tau(\chi(\tau))$ has no fully colored $(n-1)$-simplexes
(since we defined $c_\tau(\chi(\tau))$ using Lemma~\ref{lemma-color-sa}).
To finally conclude that $\langle \chi^\ell(\sigma), \sigma, val \rangle$ is locally solvable in one round,
we argue $c_\tau$ is a Sperner coloring,
which essentially follows becase \textit{val} is a Sperner-valency coloring and
$c_\tau$ is consistent and complete w.r.t. \textit{val}.
To prove the claim in detail, consider any vertex $v \in \chi^\ell(\sigma)$.
If $v \notin \chi(\tau)$, then $c_\tau(v) = id(v) \in carr(v, \sigma, \chi^\ell)$.
Otherwise, let $\tau' = carr(v, \tau, \chi)$.  Note that $ids(\tau') \subseteq carr(v, \sigma, \chi^\ell)$.
It follows from  Lemma~\ref{lemma-color-sa} that $c_\tau(\chi(\tau'))$ is either $ids(\tau')$ or~$\sigma$.
If $c_\tau(\chi(\tau')) = ids(\tau')$ then $c_\tau(v) \in carr(v, \sigma, \chi^\ell)$.
For the remaining case, note that $c_\tau(\chi(\tau')) = \sigma$ only if $dim(\tau') = n-1$ (hence $carr(\tau', \sigma, \chi^\ell) = \sigma$),
or $dim(\tau') = n-2$ and $carr(\tau', \sigma, \chi^\ell) = \sigma$
(i.e. $\tau'$ is not the chosen $(n-2)$-face of $\tau$ in the case (2.b) of Lemma~\ref{lemma-color-sa});
in either case we have that $c_\tau(v) \in carr(v, \sigma, \chi^\ell)$.
We conclude that $c_\tau$ is a Sperner coloring.
\end{proof}



%
%

\subsection{Weak Symmetry Breaking}
\label{sec:WSBlocalSolving}

\begin{theorem}
\label{theo-colorings-wsb}
For any $n\geq 3$ and  $\ell \geq  1$,
the weak symmetry breaking valency task ${\cal T}=\langle \chi^{\ell}(\sigma), \sigma, val \rangle$ in Definition~\ref{def:locSolvWSB}
is locally solvable in one round.
\end{theorem}

The proof of Theorem~\ref{theo-colorings-wsb} is similar in structure to
the proof for set agreement in the previous section.
It relies on Lemma~\ref{lemma-color-wsb} below
to produce binary colorings that are almost symmetric on the
boundary and do not have monochromatic $dim(\sigma)$-simplexes.
Figure~\ref{fig-example-wsb} shows an example of such a coloring.
In the proof of Theorem~\ref{theo-colorings-wsb}, we use these binary colorings to
locally solve symmetric binary-valency tasks.

\begin{figure}[t]
\begin{center}
\vspace{0.4cm}
\includegraphics[scale=0.4]{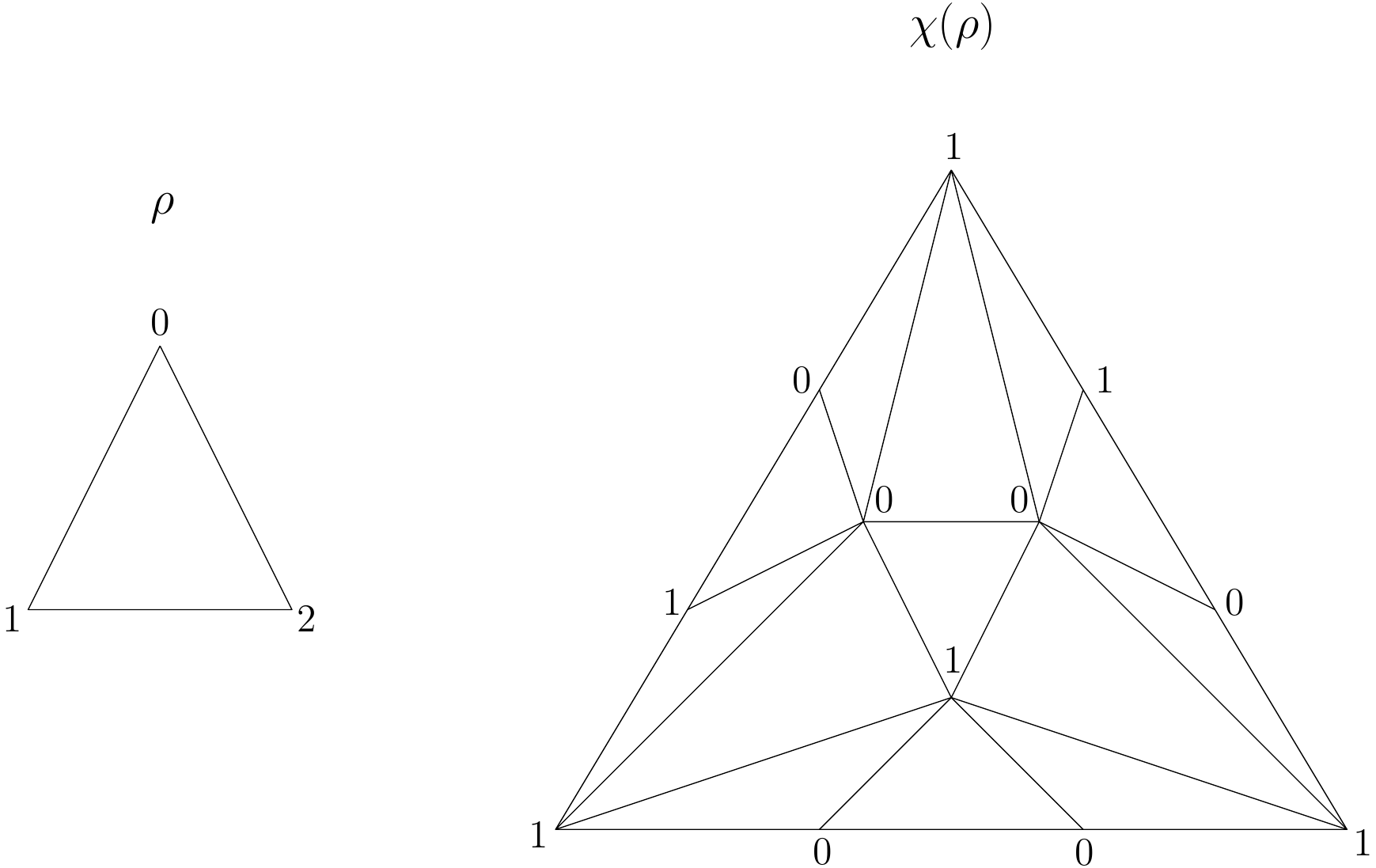}
\caption{Example for Lemma~\ref{lemma-color-wsb}.
Since every vertex $v$ of $\rho$ has dimension $n-3$, the only vertex in $\chi(v)$ has color $1$, implying (1).
For every $(n-2)$-face  $\rho'$ of $\rho$, there are vertices in $\chi(\{0,2\})$ with color 0 and 1, implying (2).
Finally, $\chi(\rho)$ has no monochromatic $(n-1)$-simplex, implying (3).}
\label{fig-example-wsb}
\end{center}
\end{figure}

\begin{lemma}
\label{lemma-color-wsb}
Consider the $(n-1)$-dimensional simplex $\rho = \{0, \hdots, n-1\}$ with $n \geq 3$.
There is a binary coloring (simplicial map) $b: \chi(\rho) \rightarrow \{0,1\}$ such that:
\begin{enumerate}
\item for every $\rho' \subset \rho$ with $dim(\rho') \leq n-3$, $b(\chi(\rho')) = \{1\}$,
\item for every $(n-2)$-dimensional face $\rho' \subseteq \rho$, $b(\chi(\rho')) = \{0,1\}$,
\item $b(\chi(\rho)) = \{0,1\}$ and there is no monochromatic $(n-1)$-simplex in $b(\chi(\rho))$.
\end{enumerate}
\end{lemma}

\begin{proof}
We exhibit such a binary coloring $b$.
For every $\rho' \subset \rho$ with $dim(\rho') \leq n-3$,
for every $v \in \chi(\rho')$, color $v$ with  1. Thus, $b(\chi(\rho')) = \{1\}$.
Let $\rho_i$ be the $(n-2)$-dimensional face of $\rho$ without vertex $i$.
We have already defined the coloring of the boundary of $\chi(\rho_i)$,
with all its vertices having color 1.

The vertices that remain to be colored are the vertices of $\tau_i$,
the ``central'' $(n-2)$-simplex of $\chi(\rho_i)$ that does not intersect
the boundary of $\chi(\rho_i)$.
For $i = 0$, color every vertex of $\tau_i$ with 0,
while for $i$, $1 \leq i \leq n-1$,
{color one vertex of $\tau_i$ with 1 and the rest with 0.}
We have:
\begin{itemize}
\item For $i=0$, $\chi(\rho_i)$ has no 1-monochromatic $(n-2)$-simplex and has exactly one
0-monochromatic $(n-2)$-simplex, which is precisely~$\tau_i$.
\item For $1 \leq i \leq n-1$, $\chi(\rho_i)$ has no 0-monochromatic $(n-2)$-simplexes
and has at least one 1-monochromatic $(n-2)$-simplex.
\item For $0 \leq i \leq n-1$, $b(\chi(\rho_i)) = \{0,1\}$.
\end{itemize}
We have defined the coloring of the boundary of $\chi(\rho)$,
and the vertices that remain to be colored are
the vertices of $\lambda$, the ``central'' $(n-1)$-simplex
$\chi(\rho)$ that does not intersect the boundary of $\chi(\rho)$.
Color with 1 the vertex of $\lambda$ with id 0;
color with 0 the remaining vertices of $\lambda$.

We now argue that there is no monochromatic $(n-1)$-simplex in $\chi(\rho)$.
First, by construction, $\lambda$ is not monochromatic.
Second, any other $(n-1)$-simplex of $\chi(\rho)$ has the form
$\gamma' \cup \gamma''$ with:
\begin{itemize}
\item $\gamma'$ being a $dim(\rho')$-simplex of $\chi(\rho')$ for a non-empty proper face $\rho'$ of  $\rho$,
\item $\gamma''$ being a non-empty proper face of  $\lambda$, and
\item $ids(\gamma') \cap ids(\gamma'') = \emptyset$.
\end{itemize}

We consider two cases.
First, if $dim(\rho') \leq n-3$, then $b(\gamma') = \{1\}$
because $b(\chi(\rho')) = \{1\}$, by definition.
Note that $dim(\gamma'') \geq 2$, and since $\lambda$ has exactly one vertex with color 1,
$b(\gamma'') = \{0,1\}$, which implies that $\gamma' \cup \gamma''$ is not monochromatic.

Otherwise, $dim(\rho') = n-2$, and hence,
$\rho'$ is one of the $\rho_i$ defined above,
namely, an $(n-2)$-dimensional face of $\rho$.
Thus, $\gamma''$ is a vertex of $\lambda$. Moreover, $ids(\gamma')$ is precisely $\rho_i$, and
hence $\gamma''$ is the vertex of $\lambda$ with id $i$, as  $ids(\gamma') \cap ids(\gamma'') = \emptyset$.
As observed above, if $i=0$, $\chi(\rho_i))$ has no 1-monochromatic $(n-2)$-simplex and has exactly one
0-monochromatic $(n-2)$-simplex; thus $\gamma'$ is not 1-monochromatic.
By construction, $\gamma''$ has color 1, and then $\gamma' \cup \gamma''$ is not 0-monochromatic.
If $1 \leq i \leq n$, $\chi(\rho_i)$ has no 0-monochromatic $(n-2)$-simplexes
and has at least one 1-monochromatic $(n-2)$-simplex; then $\gamma'$ is not 0-monochromatic.
By construction, $\gamma''$ has color 0, and then $\gamma' \cup \gamma''$ is not monochromatic.
The lemma follows.
\end{proof}

\begin{proof}[Proof of Theorem~\ref{theo-colorings-wsb}]
We now show that the weak symmetry breaking valency task
${\cal T}=\langle \chi^{\ell}(\sigma), \sigma, val \rangle$
is locally solvable in one round.
We first check that $val$ is a carrier map.
For any two simplexes $\tau', \tau \in \chi^{\ell}(\sigma)$ such that $\tau' \subseteq \tau$,
from the definition of $val$ we have that if $val(\sigma) = \{0,1\}$, then $val(\tau')$ is either
$\{0,1\}$ or $\{1\}$, and if $val(\sigma) = \{1\}$, $val(\tau')$ has to be $\{1\}$.
In any case, $val(\tau') \subseteq val(\tau)$.
We now argue that $val$ is symmetric.
For any two distinct faces $\sigma', \sigma''$ of $\sigma$
with the same dimension,
$\tau \in \chi^{\ell}(\sigma')$ and $\phi(\tau) \in \chi^{\ell}(\sigma'')$
have the same valency as they have the same dimension,
where $\phi$ is the simplicial bijection between
$\chi^{\ell}(\sigma')$ and $\chi^{\ell}(\sigma'')$ that maps
vertices preserving order.
Thus, we conclude that $\langle \chi^\ell(\sigma), \{0,1\}, val \rangle$ is a symmetric binary-valency task.
We need to show that for every input simplex $\tau \in \chi^\ell(\sigma)$,
we define a symmetric binary coloring $b_\tau$ of $\chi^{\ell+1}(\sigma)$ that is consistent and complete
w.r.t. \textit{val} and has no monochromatic $(n-1)$-simplexes in $b_\tau(\chi(\tau))$.
We focus on the case is when $\tau$ is of dimension $n-1$ because
for any simplex $\tau'$ of a smaller dimension,
we can just pick any $\tau$ containing
$\tau'$, and set $b_{\tau'}$ to
$b_\tau$ restricted to $\chi(\tau')$, i.e. $b_\tau |_{\chi(\tau')}$.

For the rest of the proof fix an $(n-1)$-dimensional simplex of $\tau \in \chi^{\ell}(\sigma)$.
We define a symmetric binary coloring $b_\tau$
that is consistent and complete w.r.t. \textit{val} and has no monochromatic $(n-1)$-simplexes in $\chi(\tau)$.
First, we use Lemma~\ref{lemma-color-wsb} to define $b_\tau$ restricted to $\chi(\tau)$, i.e. $b_\tau |_{\chi(\tau)}$,
and then extend the coloring to all vertices in $\chi^{\ell+1}(\sigma)$, to finally obtain $b_\tau$.

Let $\rho = ids(\tau)$. Note that $\rho = \sigma$ but for clarity we use $\rho$.
id's naturally induce a bijection between $\rho$ and $\tau$,
and $\chi(\rho)$ and $\chi(\tau)$,
hence any coloring (simplicial map) $\chi(\rho) \rightarrow \rho$ induces a coloring $\chi(\tau) \rightarrow ids(\tau)$.
Below, when we use Lemma~\ref{lemma-color-wsb} applied to $\rho = ids(\tau)$, we can speak about faces of $\tau$ instead of faces of $\rho$.

First, we set $b_\tau |_{\chi(\tau)}$ using a coloring of $\chi(\tau)$ in Lemma~\ref{lemma-color-wsb}.
We have that $b_\tau |_{\chi(\tau)}$ is consistent and complete with respect to \textit{val}:
for every face $\tau'$ of $\tau$, if $dim(\tau') \leq n-3$, $val(\tau') = \{1\}$, by definition of \textit{val},
and $b_\tau(\chi(\tau')) = \{1\}$, by Lemma~\ref{lemma-color-wsb}(1);
and if $dim(\tau') \geq n-2$, $val(\tau') = \{0,1\}$, by definition of \textit{val},
and $b_\tau(\chi(\tau')) = \{0,1\}$, by Lemma~\ref{lemma-color-wsb}(2-3).

We extend $b_\tau$ in two steps. In the first step,
we pick any vertex $v \in \chi^{\ell+1}(\sigma)$ that does not belong to $\chi(\tau)$ (which is uncolored yet).
If there are faces $\sigma', \sigma''$ of $\sigma$ with the same dimension such that
$v \in \chi^{\ell+1}(\sigma')$, $\phi(v) \in \chi^{\ell+1}(\sigma'')$ and $\phi(v) \in \chi(\tau)$,
where $\phi$ is the simplicial bijection between
$\chi^{\ell+1}(\sigma')$ and $\chi^{\ell+1}(\sigma'')$ that maps
vertices preserving order,
then set $b_\tau(v) = b_\tau(\phi(v))$;
otherwise, set $b_\tau(v)=1$.
In words: if $\chi(\tau)$ ``touches'' the boundary of $\chi^{\ell+1}(\sigma)$,
we replicate that ``part'' of the coloring in its symmetric ``counterparts'' in the boundary.
Observe that $b_\tau$ is well defined because, since $\ell \geq 1$, there are no two vertices of $u, v \in \chi(\tau)$
such that there are two distinct faces $\sigma', \sigma''$ of $\sigma$ of same dimension such that $u \in \chi^{\ell+1}(\sigma')$
and $v \in \chi^{\ell+1}(\sigma'')$;
intuitively, $\chi(\tau)$ can ``touch'' either $\chi^{\ell+1}(\sigma')$ or $\chi^{\ell+1}(\sigma'')$ but not both.
Note that $b_\tau$ is symmetric.

It is not hard to see that for any input simplex $\lambda \in \chi^\ell(\sigma)$ with $dim(\lambda) \leq n-3$,
$b_\tau(\chi(\lambda)) = val(\lambda)$.
First, if $\lambda$ is a face of $\tau$,
we have already saw that this is true.
Second, if $\lambda$ is a face of $\tau$, then $b_\tau(\chi(\lambda)) = \{1\}$ because even if
a vertex $v \in \chi(\lambda)$ is in the boundary of $\chi^{\ell+1}(\sigma)$ and
gets its color from a vertex $u \in \chi(\tau)$ (i.e. $b_\tau(v) = b_\tau(u)$),
it must be that  $b_\tau(u) = 1$ because there must be a face $\tau'$ of $\tau$ of dimension $dim(\lambda)$
such that $u \in \chi(\tau')$, and by Lemma~\ref{lemma-color-wsb}(1), $b_\tau(\chi(\tau')) = \{1\}$;
and finally, by definition, $val(\lambda) = \{1\}$.

However, we cannot say that same for any input simplex
$\lambda \in \chi^\ell(\sigma)$ with $dim(\lambda) \geq n-2$.
Consider the case that $\chi(\lambda)$ does not intersect $\chi(\tau)$
and the boundary of $\chi^{\ell+1}(\sigma)$;
in that case $b_\tau(\chi(\lambda)) = 1$, by definition of $b_\tau$,
but $val(\lambda) = \{0,1\}$, by definition of \textit{val}.
We fix this issue in the second step of the construction:
for any $\lambda \in \chi^\ell(\sigma)$ with $dim(\lambda) = n-2$ and $b_\tau(\chi(\lambda)) = \{1\}$,
pick the vertex $v \in \chi(\lambda)$ with smallest id among the vertices with $carr(v, \lambda, \chi) = \lambda$
(namely, $v$ is a vertex with smallest id of the ``central'' $(n-2)$-simplex of $\chi(\lambda)$),
and set $b_\tau(v) = 0$.

By construction, we have that $b_\tau(\chi(\lambda)) = \{0,1\} = val(\lambda)$.
 Note that $\lambda$ is not face of $\tau$ because initially we had $b_\tau(\chi(\lambda)) = \{1\}$,
 which is not true for $(n-2)$-dimensional faces of $\tau$, by Lemma~\ref{lemma-color-wsb}(2),
 and if $\lambda$ intersects $\tau$, then $v \notin \chi(\tau)$ because $v$ is an ``internal'' vertex
 of $\chi(\lambda)$. Therefore, $b_\tau |_{\chi(\tau)}$ remains the same after the second step.
 Moreover, since we pick vertices with smallest id, $b_\tau$ remains symmetric.
 Finally, for any  $\lambda \in \chi^\ell(\sigma)$ with $dim(\lambda) = n-1$,
 if $\lambda = \tau$, we know already that $b_\tau(\chi(\lambda)) = \{0,1\}$,
 and if $\lambda \neq \tau$, we already saw that for every $(n-2)$-face $\tau'$ of $\lambda$,
 $b_\tau(\chi(\tau')) = \{0,1\}$, and thus $b_\tau(\chi(\lambda)) = \{0,1\}$.
 Therefore, we conclude that $b_\tau(\chi(\lambda)) = val(\lambda)$, for
 every input simplex $\tau \in \chi^\ell(\sigma)$.

Thus, we have shown that $b_\tau$ is a symmetric binary coloring
that is consistent and complete w.r.t. \textit{val}.
Also, by Lemma~\ref{lemma-color-wsb},
$b_\tau(\chi(\tau))$ has no monochromatic simplexes of dimension $n-1$.
Therefore, $\langle \chi^\ell(\sigma), \{0,1\}, val \rangle$
is locally solvable in one round.
\end{proof}

\subsection{{From Weak Symmetry Breaking to $(2n-2)$-Renaming}}

{We} use the known reduction {between} weak symmetry breaking
{and} $(2n-2)$-renaming and the results in the previous subsection
to derive $(2n-2)$-renaming valency tasks that are locally solvable.

In the reduction from WSB to $(2n-2)$-renaming~\cite{GRH06},
{called ${\cal A}$ below,}
all processes first invoke an instance of weak symmetry breaking,
then the processes that obtain $b \in \{0,1\}$ invoke an independent
renaming protocol in the IIS model to obtain their output names.
{Our main observation is that} the input of each process in $\cal A$
is a binary value (together with its ID),
representing the value the process gets from weak symmetry breaking.
From this perspective,
the protocol complex of $\cal A$ is nothing else than a subdivision,
for any input simplex (i.e. a binary process assignment).
The high-level idea of our proof is to ``glue'' the local solutions of
the weak symmetry breaking valency tasks in Definition~\ref{def:locSolvWSB}
and the protocol complex of $\cal A$
to derive locally solvable valency tasks for $(2n-2)$-renaming.

{We now describe the reduction in more detail} and prove a result
that will be useful for deriving the renaming valency tasks.

{In $\cal A$, the processes that obtain $b \in \{0,1\}$ from weak
symmetry breaking invoke an independent instance of \emph{adaptive
$(2p-1)$-renaming} in the IIS model, e.g.~\cite{AttiyaWelchBook, BorowskyG1993};
we call this instance ${\cal R}_b$.
A process obtaining name $x$ from ${\cal R}_1$ decides on name $x$, while
a process obtaining name $x$ from ${\cal R}_0$ decides on name $2n-1-x$.
An adaptive renaming ensures that the size of the output name space
depends on the number of processes participating in the algorithm.
Specifically, a $(2p-1)$-adaptive renaming ensures that if $p$ processes
participate in ${\cal R}_b$, then a process returns a unique name in
$\{1, 2, \hdots, 2p-1\}$.
Weak symmetry breaking guarantees that not all processes return with
the same binary output, implying that
$p \in \{1, 2, \hdots, n-1\}$ processes participate ${\cal R}_1$ and
$n-p \in \{1, 2, \hdots, n-1\}$ processes participate ${\cal R}_0$.
It follows that the processes participating in ${\cal R}_1$
decide names in the range $\{1, 2, \hdots, 2p-1\}$,
while the processes participating in ${\cal R}_0$ decide names in the range
$\{2n-1-1, 2n-1-2, \hdots, 2n-1-(2(n-p)-1)\} = \{2n-2, 2n-3, \hdots, 2p\}$.
Therefore, $\cal A$ solves $(2n-2)$-renaming.

The next claim relies on the completeness of any IIS algorithm}
for $(2p-1)$-adaptive renaming:

\begin{claim}
\label{claim-complete}
Let $\cal R$ be {an  $(2p-1)$-adaptive renaming algorithm in
any wait-free read/write model} for $n$ processes.
For every $p \in \{1, \hdots, n\}$ and $d \in \{1, 2, \hdots, 2p-1\}$,
there is an execution of $\cal R$ with $p$ participating processes
in which a process decides name $d$.
\end{claim}

\begin{proof}
{Assume, by way of contradiction, that} there are
$p \in \{1, \hdots, n\}$ and $d \in \{1, 2, \hdots, 2p-1\}$ such that
in all executions of $\cal R$ with $p$ participating processes,
no process decides name $d$.
Note that in every solo execution only one process participate
and hence it decides~$1$.
Thus, $d \neq 1$.

Let $n'$ be the smallest integer such that $d \leq 2n'-1$.
Observe that $n' \leq p$ and $d \in \{2n'-2, 2n'-1\}$.
Without loss of generality, suppose that $d = 2n'-1$.
Thus, in every execution of $\cal A$ with $p' < n'$ participating processes,
processes decide names in $\{1, 2, \hdots, 2p'-1\}$,
while in every execution of $\cal A$ with $n'$ participating processes,
$\{1, 2, \hdots, 2p'-2\}$.
{Namely, the algorithm solves $(2p'-\lceil \frac{p'}{n'-1} \rceil)$-adaptive renaming.}
{Since {$(2p'-\lceil \frac{p'}{n'-1} \rceil)$-adaptive renaming} is equivalent to $(n-1)$-set
agreement~\cite{GMRT09}, which is not solvable
in any wait-free read/write model~\cite{BG93b,HS99,SaksZ2000},
it follows that {$(2p'-\lceil \frac{p'}{n'-1} \rceil)$-adaptive renaming} is not solvable in any wait-free
read/write model, which is a contradiction.}
\end{proof}

{We now} derive $(2n-2)$-renaming valency tasks that are locally solvable.

Fix {two integers} $n \geq 3$ and $\ell \geq 2$,
and consider the weak symmetry breaking valency task
${\cal T}=\langle \chi^{\ell}(\sigma), \{0,1\}, val \rangle$
(Definition~\ref{def:locSolvWSB}).
Let $\tau$ be any $dim(\sigma)$-dimensional simplex of $\chi^{\ell}(\sigma)$.
By Theorem~\ref{theo-colorings-wsb}, the valency task is locally solvable in one round, and hence
there is a symmetric binary decision map $b_\tau: \chi^{\ell+1}(\sigma) \rightarrow \{0,1\}$
such that $b_\tau(\chi(\tau))$ does not have
monochromatic simplexes of dimension $dim(\sigma)$.
Below, we consider the map  $b_\tau$ defined in the proof of Theorem~\ref{theo-colorings-wsb}.

Observe that any simplex $\rho \in \chi(\tau) \subset \chi^{\ell+1}(\sigma)$
is an input simplex for $\cal A$ since each process
$v \in \rho$ has binary input $b_\tau(v)$.
As already explained,
the protocol complex of $\cal A$ with input complex $\chi(\tau)$ is
$\chi^m(\chi(\tau)) = \chi^{m+1}(\tau) \subset \chi^{\ell+m+1}(\sigma)$,
where $m$ is the number of rounds $\cal A$ runs to complete.\footnote{
    Since $\chi(\tau)$ is finite, we can assume, without loss of generality,
    that all processes execute exactly $m$ rounds and then stop.}
Let $c_\tau$ denote the decision map of $\cal A$ on $\chi^{m+1}(\tau)$.
Then, for every face $\gamma$ of $\tau$,
$c_\tau(\chi^{m+1}(\gamma))$ denotes the decision set of $\cal A$
in the subcomplex  $\chi^{m+1}(\gamma) \subseteq \chi^{m+1}(\tau)$;
namely, $d \in c_\tau(\chi^{m+1}(\gamma))$ if and only if
there is a vertex $v \in \chi^{m+1}(\gamma)$ with $c_\tau(v) = d$.

\begin{lemma}
The following holds:
\begin{enumerate}
\item For every $dim(\sigma)$-dimensional $\tau \in \chi^{\ell}(\sigma)$,
$c_\tau(\chi^{m+1}(\tau)) \subseteq \{1, 2, \hdots, 2n-2\}$.

\item For every $dim(\sigma)$-dimensional $\tau \in \chi^{\ell}(\sigma)$,
$c_\tau(\chi^{m+1}(\gamma)) \subseteq c_{\tau}(\chi^{m+1}(\gamma'))$,
where $\gamma, \gamma'$ are faces of $\tau$ such that $\gamma \subset \gamma'$.

\item For every $\gamma \in \chi^{\ell}(\sigma)$, $c_\tau(\chi^{m+1}(\gamma)) = c_{\tau'}(\chi^{m+1}(\gamma))$,
where $\tau, \tau' \in \chi^{\ell}(\sigma)$ are both of dimension $dim(\sigma)$ and $\gamma$ is face of each of them.

\item For every $\gamma, \gamma' \in \chi^{\ell}(\sigma)$ of same dimension,
$c_\tau(\chi^{m+1}(\gamma)) = c_{\tau'}(\chi^{m+1}(\gamma'))$,
where $\tau, \tau \in \chi^{\ell}(\sigma)$ are both of dimension $dim(\sigma)$ and
$\gamma$ is face of $\tau$ and $\gamma'$ is face of $\gamma'$.
\end{enumerate}
\end{lemma}

\begin{proof}
The first two claims are {immediate}.
Since $\cal A$ solves $(2n-2)$-renaming,
all outputs names are in $\{1, 2, \hdots, 2n-2\}$,
and hence $c_\tau(\chi^{m+1}(\tau)) \subseteq \{1, 2, \hdots, 2n-2\}$.
Also, because $\gamma \subset \gamma'$, we have $\chi^{m+1}(\gamma) \subset \chi^{m+1}(\gamma')$, and then
$c_\tau(\chi^{m+1}(\gamma)) \subseteq c_{\tau}(\chi^{m+1}(\gamma'))$.

For the next two claims, consider first the case $dim(\gamma), dim(\gamma') \leq n-3$.
By Definition~\ref{def:locSolvWSB}, $val(\gamma) = val(\gamma') = \{1\}$.
For the third claim, we have that $b_\tau$ and $b_{\tau'}$ are consistent and complete w.r.t. $val$.
Thus, all simplexes in $b_\tau(\chi(\gamma))$ and $b_{\tau'}(\chi(\gamma))$ are 1-monochromatic,
and hence both denote the same input complex to $\cal A$, which then implies that
$c_\tau(\chi^{m+1}(\gamma)) = c_{\tau'}(\chi^{m+1}(\gamma))$.
{A similar argument proves the fourth claim}.

We {prove} the third and fourth claims {in the remaining cases,
where} $n-2 \leq dim(\gamma), dim(\gamma') \leq n-1$.

Consider first the case $dim(\gamma), dim(\gamma') = n-1$.
For the third claim we have $\gamma = \tau = \tau'$.
Then, $b_\tau = b_{\tau'}$, and consequently $b_\tau(\chi(\gamma))$ and $b_{\tau'}(\chi(\gamma))$
are the same input complex to $\cal A$, which then implies that
$c_\tau(\chi^{m+1}(\gamma)) = c_{\tau'}(\chi^{m+1}(\gamma))$.
For the fourth claim, we have $\gamma = \tau$ and $\gamma' = \tau'$.
Since $b_\tau$ is consisten and complete w.r.t. $val$ and locally solves weak symmetry breaking,
$b_\tau(\chi(\gamma))$ has a bi-chromatic  $(n-1)$-simplex $\lambda$.
Let $1 \leq p \leq n-1$ denote the number of processes with color 1 in $\lambda$.
As explained above, in every execution of $\cal A$ with input $\lambda$, those $p$
processes invoke ${\cal R}_1$ to output names in range $\{1, 2, \hdots, 2p-1\}$ while the remaining $n-p$ processes
invoke ${\cal R}_0$ output names in range $\{2p, 2p+1, \hdots, 2n-2\}$.
By Claim~\ref{claim-complete}, for every $d_1 \in \{1, 2, \hdots, 2p-1\}$ and $d_0 \in \{2p, 2p+1, \hdots, 2n-2\}$,
among all possible executions of $\cal A$ with input $\lambda$,
there is one in which a process decides $d_1$ and there is one in which a process decides $d_0$.
Therefore, $c_\tau(\chi^{m+1}(\gamma)) = \{1, 2, \hdots, 2n-2\}$.
A symmetric argument shows that  $c_{\tau'}(\chi^{m+1}(\gamma')) = \{1, 2, \hdots, 2n-2\}$.

Finally, we consider the case $dim(\gamma), dim(\gamma') = n-2$.
{In this case,} we use the properties of the local solutions
$b_\tau$ and $b_{\tau'}$ in the proof of Theorem~\ref{theo-colorings-wsb}.

For the third claim, observe that $b_\tau(\chi(\tau))$ is a coloring obtained from Lemma~\ref{lemma-color-wsb}.
The proof of that lemma shows for each $d \in \{0,1\}$, it is always the case that
$b_\tau(\chi(\gamma))$ has a $(n-2)$-simplex $\lambda_d$ with a $d$-monochromatic $(n-3)$-face.
In every execution of $\cal A$ with input $\lambda_1$, the $n-2$ processes of $\lambda_1$ with input $1$
output a name in range $\{1, 2, \hdots, 2n-5\}$.
By Claim~\ref{claim-complete}, for every $d \in \{1, 2, \hdots, 2n-5\}$,
there is an execution of $\cal A$ with input $\lambda_1$ in which a process decides $d$.
In a similar way, we can argue that for every $d \in \{4, 5, \hdots, 2n-2\}$,
there is an execution of $\cal A$ with input $\lambda_0$ in which a process decides $d$.
Thus, $c_\tau(\chi^{m+1}(\gamma)) = \{1, 2, \hdots, 2n-2\}$.
A symmetric argument shows that  $c_{\tau'}(\chi^{m+1}(\gamma)) = \{1, 2, \hdots, 2n-2\}$.

For the fourth claim, observe that the argument above already shows that
$c_\tau(\chi^{m+1}(\gamma)) = \{1, 2, \hdots, 2n-2\}$, and the same argument
with $\lambda'$ and $\tau'$ gives that $c_{\tau'}(\chi^{m+1}(\gamma')) = \{1, 2, \hdots, 2n-2\}$.
The lemma follows.
\end{proof}

A direct consequence of the previous lemma that the following is a well-defined valency task for $(2n-2)$-renaming:
${\cal R}_{\cal T}=\langle \chi^{\ell}(\sigma), \{1,2, \hdots, 2n-2\}, val \rangle$ with
$val(\gamma) = c_\tau(\chi^{m+1}(\gamma))$, for each $\gamma \in \chi^{\ell}(\sigma)$,
where $\tau$ is any $dim(\sigma)$-dimensional simplex of $\chi^{\ell}(\sigma)$
containing $\gamma$.

We now argue that ${\cal R}_{\cal T}$ is locally solvable in
$m+1$ rounds.
Consider any $dim(\sigma)$-dimensional input simplex $\tau \in \chi^{\ell}(\sigma)$ of ${\cal R}_{\cal T}$.
Recall that $c_\tau$ denotes  $\cal A$'s decision map on $\chi^{m+1}(\tau) \subset \chi^{\ell+m+1}(\sigma)$.
Since $\cal A$ solves $(2n-2)$-renaming, all decisions in $c_\tau(\chi^{m+1}(\tau))$ are correct, i.e.,
output names are in the range $\{1, 2, \hdots, 2n-2\}$ and no pair of vertices in the same simplex decide
the same output name.
Furthermore, $c_\tau$ is consistent and complete w.r.t. the valencies in $val|_\tau$,
since $val(\gamma') = c_\tau(\chi^{m+1}(\gamma'))$, by definition of $val$.
To conclude the proof,
we extend $c_\tau$ to all $\chi^{\ell+m+1}(\sigma)$, ensuring that
it is symmetric and consistent and complete w.r.t. $val$
(in similar way this is done in the proof of Theorem~\ref{theo-colorings-wsb} for weak symmetry breaking).
Therefore, ${\cal R}_{\cal T}$ is locally solvable, since the construction is independent of $\tau$.

\begin{theorem}
\label{theo-colorings-ren}
For any $n\geq 3$,
there are locally solvable valency tasks for $(2n-2)$-renaming.
\end{theorem}

%% file: FLP.tex
\section{Local Valency Impossibility Proofs}
\label{sec:locValImpProof}

Here we make precise our notion of ``impossibility proof in the FLP style,''
and use Theorems~\ref{theo-colorings-sa} and~\ref{theo-colorings-wsb}
to argue that such impossibility proofs do not exist for
$(n-1)$-set agreement and weak symmetry breaking in the IIS model.

In a \emph{local valency impossibility proof} for say, set agreement,
one assumes by way of contradiction a hypothetical $R$-round protocol
solving the task. Recall that the protocol determines valencies,
for all simplexes in all rounds, starting with those of the initial configuration $\sigma$.
The valencies must respect the task specification, since we asume the protocol solves the task.
For example, $val(P_i,i)=\{ i\}$, where $(p_i,i)\in\sigma$ is the initial state of $P_i$ (in an execution where $P_i$
sees only itself, it must decide its own input value).
A crucial observation is that what we are given in a local valency impossibility proof are
only the valencies, and there are many protocols that could produce the same valencies (i.e.,
many simplicial  maps $c$ assigning decisions to $\chi^R(\sigma)$, yielding the same valencies).

The proof consists of $R-1$ phases to select  a sequence of simplexes
$\sigma_0,\sigma_1,\ldots,\sigma_{R-1}$, starting with $\sigma_0=\sigma$ and such that
$\sigma_\ell \in\chi^\ell(\sigma)$ for all $1 \leq \ell \leq R-1$, extending the sequence by one
at each phase.

Assume we have selected the sequence $\sigma_0,\ldots,\sigma_\ell$, for some $\ell\geq 1$.
To select $\sigma_{\ell+1}\in\chi(\sigma_{\ell}) \subset \chi^{\ell+1}(\sigma)$,
one considers all simplexes in $\sigma'\in\chi(\sigma_{\ell})$,
together with their valencies, $val(\sigma')$.
When we reach phase $R-1$, and we have selected $\sigma_{R-1}\in\chi(\sigma_{R-2})$,
the protocol reveals all decisions in $\chi(\sigma_{R-1}) \subset \chi^R(\sigma)$
(and \emph{only} those decisions).
Namely, a simplicial map $c$ assigning a decision to each vertex of $\chi(\sigma_{R-1})$,
respecting all previously observed valencies, namely, all those  in each $\chi(\sigma_\ell)$.

There \emph{is a local valency impossibility proof} for the task
if and only if one can select a sequence
$\sigma_0, \sigma_1, \hdots, \sigma_{R-1}$ such that
the task is \emph{not} locally solvable in one round at $\sigma_{R-1}$.
Namely, if there is no  decision function $c$,
that respects the valencies and is consistent with the task specification.
In the case of set agreement,
at least one simplex must have $n$ different decisions,
for any $c$ that respects the valencies.

Therefore, there is no such a proof if we are able to exhibit valencies
of a hypothetical protocol such that,
for any selection $\sigma_0, \sigma_1, \hdots, \sigma_{R-1}$,
there is a decision function $c$ corresponding to those valencies
that locally solves $(n-1)$-set agreement
(the argument for weak symmetry breaking is analogous):
\begin{itemize}
\item Fix any $R \geq 2$.

\item For the input complex, the valency of each
$\sigma'\subseteq\sigma$ is $val(\sigma') = \sigma'$.

\item In phase $\ell \in \{0, \hdots, R-2\}$, the valency of each simplex $\tau \in \chi(\sigma_\ell) \subset \chi^{\ell+1}(\sigma)$
is the valency of the simplex in the valency task ${\cal T}^{\ell+1}=\langle \chi^{\ell+1}(\sigma), \sigma, val \rangle$
in Definition~\ref{def:locSolvSAtask}, namely, $val(\tau)$.

\item In phase $R-1$, the protocol picks a decision map $c: \chi^R(\sigma) \rightarrow \sigma$ that is consistent
and complete w.r.t  ${\cal T}^{R-1}=\langle \chi^{R-1}(\sigma), \sigma, val \rangle$
and does not have fully colored $(n-1)$-simplexes in $\chi(\sigma_{R-1}) \subset \chi^R(\sigma)$,
and provides only the decisions $c(\chi(\sigma_{R-1}))$.
Such a mapping exists since ${\cal T}^{R-1}$  is $(n-1)$-locally solvable,
due to Theorem~\ref{theo-colorings-sa}.
\end{itemize}

Notice that no matter the simplex $\sigma_\ell$ we chose in each phase,
we cannot find a contradiction in the decisions of $\chi(\sigma_{R-1})$.
The only thing that remains to be argued is that the valencies
are consistent during all phases.
More specifically, valencies preserve containment in the same phase
and can only shrink as the phases go by,
and additionally they do not contradict validity,
i.e., the valency of a simplex is a subset of its carrier.
Thus, for any $R$, there are valencies that could be produced by
a hypothetical set agreement protocol.
This is implied by the three properties below that are satisfied for
every valency task
${\cal T}^{\ell}=\langle \chi^{\ell}(\sigma), \sigma, val \rangle$,
$\ell \in \{1, \hdots, R-1\}$,
and whose proof is based on Observation~\ref{observation-2}.
These properties also show that the decisions of
$\chi^{R-1}(\sigma_{R-1})$ revealed by the protocol are consistent
with all valencies in all phases.

\begin{observation}
\label{observation-2}
For $\ell \geq 0$, for every $\gamma \in \chi^\ell(\sigma)$, 
$ID(\gamma) \subseteq carr(\gamma, \sigma, \chi^\ell)$.
Furthermore, if $dim(carr(\gamma, \sigma, \chi^\ell)) = dim(\gamma)$,
then $carr(\gamma, \sigma, \chi^\ell) = ID(\gamma)$.
\end{observation}

\begin{description}
\item[Containment]
For $\tau, \tau' \in \chi^\ell(\sigma)$ with $\tau' \subset \tau$, we have $val(\tau') \subseteq val(\tau)$.
By Observation~\ref{observation-2}, for every $\gamma \in \chi^\ell(\sigma)$,
$ID(\gamma) \subseteq carr(\gamma, \sigma, \chi^\ell)$.
Since $\tau' \subset \tau$, we have $carr(\tau', \sigma, \chi^\ell) \subset carr(\tau, \sigma, \chi^\ell)$.
Depending on the dimension of $\tau$,
$val(\tau)$ is either $ID(\tau)$ or $carr(\tau, \sigma, \chi^\ell)$;
and similarly for $\tau'$.
Therefore, $val(\tau') \subseteq val(\tau)$.

\item[Valencies shrink]
Consider any $m > \ell$ with $0 \leq \ell + m \leq R-1$ and
the valency task ${\cal T}^{\ell+m}=\langle \chi^{\ell+m}(\sigma), \sigma, val' \rangle$.
For $\tau \in \chi^\ell(\sigma)$ and $\tau' \in \chi^m(\tau)$, $val'(\tau') \subseteq val(\tau)$.
The argument is very similar to the previous one.
Since $\tau' \in \chi^m(\tau) \subset \chi^\ell(\sigma)$,
we have that $carr(\tau', \sigma, \chi^{\ell+m}) \subseteq carr(\tau, \sigma, \chi^\ell)$.
Depending on the dimension of $\tau$, $val(\tau)$ is either $ID(\tau)$
or $carr(\tau, \sigma, \chi^\ell)$, and we have that $ID(\tau) \subseteq carr(\tau, \sigma, \chi^\ell)$, by Observation~\ref{observation-2};
and similarly for $\tau'$.
Therefore, $val'(\tau') \subseteq val(\tau)$.

\item[Validity] For $\tau \in \chi^\ell(\sigma)$, we have $val(\tau) \subseteq carr(\tau, \sigma, \chi^\ell)$.
If $dim(\tau)$ is $n-2$ or $n-1$, $val(\tau) = carr(\tau, \sigma, \chi^\ell)$,
and if $dim(\tau) \leq n-3$, $val(\tau) = ID(\tau) \subseteq carr(\tau, \sigma, \chi^\ell)$,
where the last containment follows from Observation~\ref{observation-2}.
\end{description}

The previous properties hold for all simplexes of the valency tasks in Definition~\ref{def:locSolvSAtask},
and thus we conclude that there is no contradiction on the valencies
provided during the phases.

%% file: discussion.tex

\section{Discussion}
\label{section:them}

This paper argues that the $(n-1)$-set agreement, weak symmetry breaking
and $(2n-2)$-renaming impossibilities in the wait-free read/write
shared memory model cannot be proved using local arguments,
in the style of FLP.
We introduced the notions of \emph{valency task} and \emph{local solvability}
for set agreement and weak symmetry breaking.
We formalized the notion of \emph{local-valency impossibility proof}
for these tasks, where a presumptive  protocol for {these} tasks can
always hide erroneous results, even after committing to valencies one round
before termination.
We showed that there are no local-valency impossibility proofs for $(n-1)$-set agreement
and weak symmetry breaking in the wait-free read/write shared memory model.

Alistarh, Aspnes, Ellen, Gelashvili and Zhu~\cite{AlistarhAEGZ2019}
studied a similar question by defining a game between
a \emph{prover} and a \emph{protocol},
as a way to represent extension-based techniques for proving
impossibility results.
They have shown that, for set agreement,
a protocol can win this game against any prover,
thus showing extension-based techniques do not suffice for proving
the impossibility of solving set agreement.
Their approach is restricted  to \emph{unbounded} protocols.
This also complicates the argument, since they need to work with
\emph{non-uniform} simplicial subdivisions.
In contrast, we consider \emph{bounded wait-free}.
This allows to assume that all processes decide at the same round, $R$
(hence giving more information and power to the prover),
leading to simpler \emph{uniform} subdivisions.
We stress that there is no loss of generality in this assumption,
since a task is wait-free solvable if and only if it is wait-free solvable
by a protocol where all processes decide at the same round.
Furthermore, while Alistarh et al.\ study only $k$-set agreement,
we also investigate weak symmetry breaking,
and by reduction, renaming.

{

Looking forward, one would like to have notions of valency tasks
and local impossibility arguments to be preserved under reductions 
and simulations.
A notion of valency tasks, 
which is robust with respect to implementation relations,
might be feasible for \emph{rendezvous tasks}, 
including set agreement, simplex agreement, loop agreement,
and approximation agreement~\cite{LIU09}.
A rendezvous task is said to implement another if an instance of its 
solution, followed by a wait-free protocol using read/write registers, 
solves the other.
Rendezvous tasks are divided into infinitely many, countable classes, 
where two tasks are in the same class if they implement each other~\cite{LIU09}.
Brusse and Ellen~\cite{BrusseE21} studied reductions
for \emph{augmented} extension-based proofs.}



In the context of randomized and non-deterministic protocols, 
it would also be interesting
to consider valency tasks and local impossibilities.
For example, the {lower bound proof of~\cite{AttiyaC2008} 
is based on valency and has a local nature}.

Another interesting question is \emph{how much} of the final decisions
the protocol can reveal;
for example, revealing consistent decisions even if several configurations
are chosen instead of only one.
Finally, we would like to explore local-valency
proofs beyond our wait-free setting, in
models that are not round-based~\cite{LubMoran95}
or non-compact~\cite{NSW19}, like $t$-resilient models.